\documentclass[orivec]{llncs}

\usepackage[utf8]{inputenc}
\usepackage{amsmath,amssymb,amsfonts,amscd, amsbsy,mathtools,amsxtra,mathrsfs,stmaryrd}
\usepackage{graphicx}
\usepackage{enumerate}
\usepackage{enumitem} 
\usepackage{hyperref}
\usepackage{textcomp}
\usepackage[english]{babel}
\usepackage{lineno}
\everymath{\displaystyle}
\DeclareMathAlphabet{\mathcal}{OMS}{cmsy}{m}{n}

\usepackage{MnSymbol} 
\usepackage{xcolor}
\usepackage{bussproofs}
\usepackage{proof}
\usepackage{xspace}
\usepackage{dsfont}
\usepackage{srcltx}
\usepackage{relsize}
\usepackage{stackrel}

\usepackage{algorithm,algpseudocode,algorithmicx}
\usepackage{tikz,float}
\usepackage{tikz-cd}
\usepackage{pgf,tikz}
\usepackage{mathpartir}
\usepackage{enumitem}

\usepackage[normalem]{ulem}

\newcommand{\labTriEq}[1]{\ensuremath{\triangleq_{{\tiny #1}}}}
\newcommand{\genOf}[1]{\ensuremath{\preceq_{{\tiny #1}}}}

\newcommand{\approxOf}[1]{\ensuremath{\approx_{{\tiny #1}}}}

\newcommand{\deriv}[1]{\ensuremath{\overset{{{\tiny #1}}}{\Longrightarrow}}}
\newcommand{\ddec}{\deriv{\mbox{\it Dec}}}

\newcommand{\dsolve}{\deriv{\mbox{\it Sol}}}
\newcommand{\ddsolve}[1]{\deriv{\ensuremath{\mbox{\it Sol } #1}}}

\newcommand{\dexpalf}{\deriv{\mbox{\it ExpLA1}}}
\newcommand{\dexparf}{\deriv{\mbox{\it ExpRA1}}}
\newcommand{\dexpals}{\deriv{\mbox{\it ExpLA2}}}
\newcommand{\dexpars}{\deriv{\mbox{\it ExpRA2}}}
\newcommand{\dmer}{\deriv{\mbox{\it Mer}}}

\newcommand{\mcsg}{\mathit{mcsg}}
\newcommand{\mcsgabs}{\mathit{mcsg}_{\tiny \Abs}}

\newcommand{\rdec}{(Dec)}

\newcommand{\rsolve}{(Sol)}

\newcommand{\rexpalf}{(ExpLA1)}

\newcommand{\rexpals}{(ExpLA2)}
\newcommand{\rexpars}{(ExpRA2)}
\newcommand{\rmer}{(Mer)}

\newcommand{\AU}{\textsc{AUnif}}
\newcommand{\Abs}{\mbox{\tt Abs}}
\newcommand{\Abst}{\ensuremath{\mbox{\textuparrow}}}

\newcommand{\RVar}{\mbox{\it Rvar}}

\newcommand{\Dom}{\mbox{\it Dom}}

\newcommand{\labels}[1]{\ensuremath\mbox{\it labels}(#1)}

\allowdisplaybreaks

\makeatletter\@twosidefalse\@mparswitchfalse\makeatother
\begin{document}
\title{Equational Anti-Unification over Absorption Theories}


\author{Mauricio Ayala-Rincón\inst{1} \and David Cerna\inst{2} \and
Andrés Felipe González Barragán\inst{1} \and  Temur Kutsia\inst{3}}
\institute{Universidade de Brasília\\
\email{ayala@unb.br, andres.felipe@aluno.unb.br}
\and
Czech Academy of Sciences Institute for Computer Science
\\ \email{dcerna@cs.cas.cz}
\and
RISC, Johannes Kepler University Linz\\
\email{kutsia@risc.jku.at}
}

\maketitle 
\begin{abstract} 
Interest in anti-unification, the dual problem of unification, is on the rise due to applications within the field of software analysis and related areas. For example, anti-unification-based techniques have found uses within clone detection and automatic program repair methods. While syntactic forms of anti-unification are enough for many applications, some aspects of software analysis methods are more appropriately modeled by reasoning modulo an equational theory. Thus, extending existing anti-unification methods to deal with important equational theories is the natural step forward. This paper considers anti-unification modulo pure absorption theories, i.e., some operators are associated with a special constant satisfying the axiom $f(x,\varepsilon_f) \approx f(\varepsilon_f,x) \approx \varepsilon_f$. We provide a sound and complete rule-based algorithm for such theories. Furthermore, we show that anti-unification modulo absorption is infinitary. Despite this, our algorithm terminates and produces a finitary algorithmic representation of the minimal complete set of solutions. We also show that the linear variant is finitary.
 \end{abstract}

\section{Introduction}

Anti-unification (AU) is a fundamental operation
for reasoning about generalizations of formal objects. It is the dual operation to unification. The seminal works of Plotkin and Reynolds, introducing the area, were published more than fifty years ago (\cite{Plotkin70,Reynolds70}); however, only recently has interest in the development of foundations of AU from the software and related communities gained attraction. This recent tendency is mainly due to the significance of generalization operations within frameworks crucial for software analysis and related areas \cite{DBLP:conf/ijcai/CernaK23}. In contrast to unification, where identifying the equivalence classes induced by a set of expressions is the main objective, AU methods search for the least general commonalities induced by a set of expressions. Investigations have exploited AU methods for various applications such as the implementation of efficient parallel compilers \cite{DBLP:journals/fgcs/BarwellBH18}, plagiarism detection and code cloning \cite{DBLP:conf/lopstr/VanhoofY19,DBLP:journals/tplp/YernauxV19,DBLP:conf/csl/YernauxV22}, automated bug detection and fixing \cite{DBLP:conf/sbes/SousaSGBD21,DBLP:conf/nsdi/MehtaB0BMAABK20}, and library learning/compression \cite{DBLP:journals/pacmpl/CaoKNWTP23}. 
Investigations have considered AU for several mathematical and computational frameworks such as term-graphs \cite{DBLP:conf/rta/BaumgartnerKLV18}, higher-order variants \cite{DBLP:conf/ausai/KrumnackSGK07,DBLP:journals/jar/BaumgartnerKLV17},  unranked languages (in which function symbols have variable arity) \cite{DBLP:journals/jar/KutsiaLV14,DBLP:journals/iandc/BaumgartnerK17}, nominal terms \cite{DBLP:conf/rta/BaumgartnerKLV15,DBLP:conf/fscd/Schmidt-Schauss22,DBLP:conf/cade/SchmidtSchaussN23}, approximate AU \cite{DBLP:conf/tbillc/KutsiaP19,DBLP:conf/cade/KutsiaP22,AITKACI20201},  and first-order equational AU, which is also the subject of this paper. 

In their works, Plotkin and Reynolds introduced syntactic AU algorithms for computing least general generalizations (lggs). In the equational case, the given terms do not necessarily have a single lgg; thus, problems are instead characterized by their minimal and complete sets of generalizations (mcsg's), which leads to the classification of theories depending on the existence and cardinality of such sets: If the mcsg does not exist for some problem in the given theory, then the theory has the nullary AU type. Otherwise, theories may have unitary (all problems have a singleton mcsg, i.e., a single lgg), finitary (all problems have a finite mcsg, at least one of which is not a singleton), or infinitary (there is a problem with the infinite mcsg) AU type. 

There have been quite a few developments concerned with AU modulo equational theories. For example, Burghardt~\cite{DBLP:journals/ai/Burghardt05} considered anti-unification modulo an arbitrary equational theory using grammars. 
Most other authors studied AU over fundamental algebraic properties and their combinations, e.g.,  associative $(A)$, commutative $(C)$, $AC$, idempotent  $(I)$ operators, or operators with unit ($U$) elements (e.g., $f(x,e)$). An early work by Baader \cite{DBLP:conf/rta/Baader91} studied AU over so-called ``commutative theories'', covering commutative monoids $(ACU)$, commutative idempotent monoids $(ACUI)$, and Abelian groups. In a restricted setting, he showed that AU in such theories is unitary. Alpuente et al.~\cite{DBLP:journals/iandc/AlpuenteEEM14,DBLP:journals/amai/AlpuenteEMS22} studied AU over combinations of $A$, $C$, and $U$ operators in an order-sorted setting, providing complete AU algorithms, and proved that all studied AU problems are of type finitary. A further investigation by Cerna and Kutsia~\cite{DBLP:conf/fscd/CernaK20} showed that some results depend on the number of symbols that satisfy
the associated equational axioms. For instance, they proved the nullarity of theories containing more than one equational symbol: $U^{>1}, (AU)^{>1} (CU)^{>1}, (ACU)^{1}$, and  $(AU)(CU)$. Also, these authors showed that  $I, AI$, and $CI$ are of type infinitary \cite{DBLP:journals/tocl/CernaK20}, and Cerna proved that $(UI)^{>1}, (AUI)^{>1}, (CUI)^{>1},(ACUI)^{>1}$, and semirings are of type nullary~\cite{DBLP:journals/tcs/Cerna20}. 

This paper extends the state-of-the-art on equational anti-unification by providing an algorithm to solve AU problems in a first-order syntax that includes operators with collapsing symbols, i.e., symbols that are associated with an absorption constant such that  $f(\varepsilon_f,x)\approx \varepsilon_f \approx f(x, \varepsilon_f)$.   Such collapsing properties often appear in syntactic, logical, and algebraic frameworks (e.g., $0 \times x \approx 0, \mbox{false} \wedge p \approx \mbox{false}$). They are an instance of so-called subterm-collapsing equational theories. Concerning software development and programming languages, one could consider such operations as modeling exception handling and other methods of flagging errors where much of the context is discarded when the error handling code is triggered. In such cases, like absorption theories, the state prior to triggering the error handling code is not precisely captured by the resulting context and, in a sense, can be abstracted away.

The main results presented in this paper are (i) a terminating algorithm for anti-unification over absorption theories (Section \ref{sec:RulesAndTermination}), (ii) proofs of soundness and completeness (Section \ref{sec:SoundnessAndCompleteness}), (iii) a proof that anti-unification over absorption theories is of type infinitary (Section~\ref{sec:AUtype}),
(iv) a finitary representation of a potentially infinite set of solutions (Section \ref{sec:RulesAndTermination}), and (v) a brief analysis of a finitary linear variant (Section~\ref{sect:linear}).
\section{Preliminaries}
Let  $\mathcal{V}$ be a countable set of variables and $\mathcal{F}$ a set of function symbols each associated with an arity. Additionally, we assume $\mathcal{F}$ contains a special constant $\star$, referred to as the \textit{wild card}. The set of terms derived from the sets mentioned above is denoted by  $\mathcal{T}(\mathcal{F},\mathcal{V})$, whose members are constructed using the grammar  $t ::= x \mid f(t_1,\dots,t_n)$, where $x\in\mathcal{V}$ and $f\in \mathcal{F}$ with arity $n\geq 0$. When $n=0$, $f$ is called a {\it constant}. Constant and function symbols, terms, and variables are denoted by lower-case letters of the first, second, third, and fourth quarter of the alphabet ($a,b,\ldots$;  $f,g,\ldots$; $r,s,\ldots$; $x,y,\ldots$). The set of variables occurring in $t$ is denoted by $\mathcal{V}(t)$. The {\em length} of a term is defined inductively as: $len(x)=1$, and $\textstyle len(f(t_1,\ldots,t_n)) = 1 + \sum_{i=1}^{n} len(t_i)$. 

The set of {\it positions} of a term $t$ is a set of strings over the positive integers defined in the standard way. 
The {\em head} of a term $t$ is defined as $head(x)=x$ and $head(f(t_1,\dots,t_n))=f$, for $n\geq 0$.  

A {\it substitution} is a function $\sigma: \mathcal{V} \to \mathcal{T(\mathcal{F},\mathcal{V})}$ such that $\sigma(x)\neq  x$ for only finitely many variables. The set of the variables that are not mapped to themselves is called the \emph{domain} of $\sigma$, denoted as $\Dom(\sigma)$. The \emph{range} of $\sigma$, denoted $Ran(\sigma)$, is the set of terms $\{\sigma(x) \mid x\in \Dom(\sigma) \}$. We refer to a substitution $\sigma$ as \emph{ground} if for all $t\in Ran(\sigma)$, $\mathcal{V}(t)=\emptyset$. Substitutions are extended to terms in the usual manner. We use the postfix notation for substitution application to terms and write $t\sigma$ instead of $\sigma(t)$.

Substitutions can be described as sets of {\em bindings} of variables in their domains into terms in their ranges, e.g.,  we represent a substitution $\sigma$  as the set $\{x\mapsto x\sigma \mid x\in  \Dom(\sigma) \}$.
Lower-case Greek letters denote substitutions except for the identity substitution that we denote by $id$. The set of variables occurring in $Ran(\sigma)$ is denoted as $\RVar(\sigma)$. The {\it 
composition} of substitutions $\sigma$ and $\rho$ is written  $\sigma\rho$ and 
is associative, i.e. $x(\sigma\rho)=(x\sigma)\rho$ for each $x\in \mathcal{V}$. 
The \emph{restriction of a substitution} $\sigma$ to a set of variables $V$, denoted by $\sigma\vert_{V}$, is a substitution defined as $\sigma\vert_{V}(x)=\sigma(x)$ for all $x\in V$ and $\sigma\vert_{V}(x)=x$ otherwise.

In this work, we focus on equational anti-unification. Thus, we refrain from presenting syntactic variants of the concepts discussed below. For such details, we refer to the recent survey on the topic~\cite{DBLP:conf/ijcai/CernaK23}.

\begin{definition}[Equational Theory \cite{DBLP:journals/jsc/Siekmann89}]
\label{def:equtheo}
An equational theory $T_E$ is a class of algebraic structures that hold a set of equational axioms $E$ over a  set of terms $\mathcal{T(\mathcal{F},\mathcal{V})}$.  
\end{definition}
The relation $\{ (s,t) \in \mathcal{T(\mathcal{F},\mathcal{V})} \times \mathcal{T(\mathcal{F},\mathcal{V})} \; | \;  E\models (s,t)\}$ induced by a set of equalities $E$ gives the set of equalities satisfied by all structures in the theory of $E$. We will use the notation $s\approx_E t$ for $(s,t)$ belonging to this set. Also, we will identify $T_E$ with the set of axioms $E$. Groups, monoids, and semirings are examples of equational theories.  

\begin{definition}[$E$-generalization, $\genOf{E}$, solution] 
\label{def:egen}
The generalization relation of the theory induced by $E$ holds for terms $r,s\in \mathcal{T(\mathcal{F},\mathcal{V})}$,  written $r\genOf{E} s$, if there exists a substitution $\sigma$ such that $r\sigma\approxOf{E} s$. In this case,  we say that $r$ is {\it more general} than $s$ modulo $E$. 
If $r\genOf{E}  s$ and $r\genOf{E}  t$, we say that $r$ is an $E$-generalization of $s$ and $t$. The set of all $E$-generalizations of $s$ and $t$ is denoted as $\mathcal{G}_{E}(s,t)$. By $\prec_E$ and $\simeq_E$ we denote the strict and equivalence relations induced by $\genOf{E}$. 

Each $r\in \mathcal{G}_{E}(s,t)$ is associated with two substitutions $\sigma$ and $\rho$ such that $r\sigma\approxOf{E} s$ and $r\rho\approxOf{E} t$.  The triple $\langle r, \sigma,\rho\rangle$ is the $E$-{\it solution} of $s$ and $t$ associated to the $E$-generalization $r$. 

\end{definition}





\begin{example}
\label{ex:lgg}
Consider the equational theory $\Abs=\{f(\varepsilon_f,x)\approx \varepsilon_f, f(x,\varepsilon_f)\approx\varepsilon_f\}$, and the terms $s=\varepsilon_f$ and $t=f(f(b,c),a)$. 
Then $f(f(b,x),a)$ is \Abs-generalization of $s$ and $t$.  Indeed, $\sigma=\{x\mapsto \varepsilon_f\}$ and $\rho=\{x\mapsto c\}$ satisfy $f(f(b,x),a)\sigma=f(f(b,\varepsilon_f),a)\approx_{\Abs}\varepsilon_f$ and $f(f(b,x),a)\rho=f(f(b,c),a)$. Hence, $\langle f(f(b,x),a), \sigma, \rho\rangle$ is an \Abs-solution to the problem of generalizing $s$ and $t$.

\end{example}

\begin{definition}[Minimal complete set of $E$-generalizations]
\label{def:mcsg}
The \emph{minimal complete set of $E$-generalizations} of the terms $s$ and $t$, denoted as $\mcsg_E(s,t)$,  is a subset of $\mathcal{G}_{E}(s,t)$ satisfying:
\begin{enumerate}
\item For each $r\in \mathcal{G}_{E}(s,t)$ there exists $r'\in \mcsg_E(s,t)$ such that $r\genOf{E} r'$.
\item If $r,r'\in \mcsg_E(s,t)$ and $r\genOf{E} r'$, then $r=r'$ 
(Minimality).
\end{enumerate}
\end{definition}

\begin{example}
\label{ex:mcsg}
For Example \ref{ex:lgg}, the minimal complete set of $\Abs$-generalizations is
\[\mcsgabs(\varepsilon_f,f(f(b,c),a))=\{f(f(x,c),a),f(f(b,x),a),f(f(b,c),x)\}.\]
\end{example}

\begin{definition}[Anti-unification type]
\label{def:typemcsg}
The anti-unification type of an equational theory $E$ may have one of the following forms: \begin{itemize}
\item \textit{Unitary}:  $\mcsg_{E}(s,t)$ exists for all  $s,t\in \mathcal{T(\mathcal{F},\mathcal{V})}$ and is always singleton. 
\item \textit{Finitary}: $\mcsg_{E}(s,t)$ exists and is finite, for all  $s,t\in \mathcal{T(\mathcal{F},\mathcal{V})}$, and there exist  $s',t'\in  \mathcal{T(\mathcal{F},\mathcal{V})}$ for which  $1 < |\mcsg_{E}(s',t')| < \infty$. 
\item \textit{Infinitary}: $\mcsg_{E}(s,t)$ exists for all  $s,t\in  \mathcal{T(\mathcal{F},\mathcal{V})}$, and there exist $s',t'\in  \mathcal{T(\mathcal{F},\mathcal{V})}$ such that  $\mcsg_{E}(s',t')$  is infinite. 
\item \textit{Nullary}: for some  $s,t\in   \mathcal{T(\mathcal{F},\mathcal{V})}$, $\mcsg_{E}(s,t)$ does not exist. 
\end{itemize}
\end{definition}

\begin{example}  From the introduction:
    Syntactic AU  is \textit{unitary}~\cite{Plotkin70,Reynolds70}, AU over associative (A) and commutative (C)  theories is \textit{finitary}~\cite{DBLP:journals/iandc/AlpuenteEEM14},  AU over idempotent theories is \textit{infinitary}~\cite{DBLP:journals/tcs/Cerna20}, and AU with multiple unital equations is \textit{nullary}~\cite{DBLP:conf/fscd/CernaK20}.
\end{example}

 \section{Anti-Unification in Absorption Theories}
\label{sec:RulesAndTermination}

Absorption is one of the fundamental algebraic properties used to define various algebraic structures. For example,  semirings, rings, and Boolean algebras define their multiplicative operation as absorption with respect to the additive identity. Concrete examples are the product operation and $0$ in number fields and the intersection operation and $\emptyset$ in set theory. So far, investigations on anti-unification over absorption theories have only considered equational theories defining more elaborate algebraic structures (\textit{semirings}~\cite{DBLP:journals/tcs/Cerna20}). In this work, we study pure absorption theories as part of a general study on the anti-unification of subterm-collapsing theories. 

Given a binary function symbol $f$ and a constant $\varepsilon_f$, the absorption property is given by the axioms $\Abs(f,\varepsilon_f)= \{f(x,\varepsilon_f)\approx\varepsilon_f,  f(\varepsilon_f,x)\approx\varepsilon_f\}$. An absorption theory is a finite union of $\Abs(f_1,\varepsilon_{f_1}), \ldots,\Abs(f_n,\varepsilon_{f_n})$ such that the chosen symbols are pairwise distinct. We denote such theories by $\Abs$.

We refer to symbols contained in the equations of $\Abs(f,\varepsilon_f)$ \emph{related absorption symbols}, and  $\varepsilon_f$ as the \emph{absorption constant of $f$}.  

For the rest of this paper, we will only consider anti-unification problems over $\Abs$ theories, i.e., purely absorption equational theories. An \emph{anti-unification equation} (AUE) is a triple of the form $s\labTriEq{x} t$, where $x\in \mathcal{V}$, called the \emph{label of the AUE}, and $s,t\in \mathcal{T(\mathcal{F},\mathcal{V})}$. Given a set $A$ of  AUEs,  $labels(A) = \{ x\mid s\labTriEq{x} t\in A\}$. A set of AUEs is \emph{valid} if its labels are pairwise disjoint. We extend the notion of $length$ to AUEs and sets of AUEs as the sum of the lengths of the terms in the AUEs.
The \emph{wild card} plays an integral role in our procedure for computing generalizations. In particular, an AUE is referred to as \emph{wild} if either the left or right side is the wild card. The procedure aims to compute a set of terms generalizing the input AUE and a set of \emph{solved} AUEs from which we can compute how such terms generalize the input AUE.

\begin{definition}[Solved AUE]
An AUE $s\labTriEq{x} t$ is \emph{solved} over an absorption theory $\Abs$ if $head(s)$ $\not = head(t)$, $head(s)$ and $head(t)$ are not related absorption symbols, and $s\labTriEq{x} t$ is not wild.
\end{definition}

\subsection{Generalization Procedure for \Abs\ Theories}
We now introduce a procedure for solving the anti-unification problems over \Abs\ theories. It is presented as the set of inference rules \AU\ in Table \ref{tab:generalizationrules}, which transform quadruples, called \textit{configurations}, defined below.

\begin{definition}[Configuration]
\label{def:configuration} A configuration is a quadruple of the form $ \langle A;S;T;\theta\rangle$, where:
\begin{itemize}
    \item $A$ is the  valid set of {\em unsolved} AUEs;
    \item $S$ is the {\em store}, the valid set of {\em solved} AUEs;  
    \item $T$ is the {\em abstraction}, the valid set of {\em wild} AUEs;
    \item $\theta$ is a {\em substitution} mapping the labels to their respective generalizations.
\end{itemize}
All terms contained in a configuration are in their $\Abs$-normal forms: an absorption constant does not occur as the argument to its absorption symbol. Configurations
satisfy the following properties:
  
\begin{enumerate}[label=(\roman*)]\label{list:propconfig}
 \item The sets $\labels{A},\labels{S},\labels{T}$ and $\Dom(\theta)$ are pairwise disjoint.
 \item $\RVar(\theta)=\labels{A}\cup \labels{S}\cup \labels{T}$. 
\end{enumerate}
\end{definition}

 
 The rules in the Table~\ref{tab:generalizationrules} will be referred to as follows:  Decompose (\ddec), Solve (\dsolve),  Expansions for Left Absorption, (\dexpalf and \dexpals), Expansions for Right Absorption (\dexparf and \dexpars), and Merge (\dmer). 

\begin{table}
\centering
\caption{Generalization \AU\ rules for \Abs\ theory.}\label{tab:generalizationrules}
\begin{tabular}{|@{\hspace{2mm}}l@{\hspace{3mm}}c@{\hspace{2mm}}|}
\hline
& \\[0mm] 
\textbf{(\ddec)}
&
$\displaystyle\frac{\langle \{f(s_1,\dots,s_n)\labTriEq{x} f(t_1,\dots,t_n)\}\cupdot A;S;T;\theta\rangle}{\langle \{s_1\labTriEq{y_1} t_1, \dots, s_n\labTriEq{y_n} t_n\}\cup A; S;T; \theta\{x\mapsto f(y_1,\dots,y_n)\}\rangle}$\\
& where $f$ is an $n$-ary symbol, $n\geq 0$, and $y_1,\dots,y_n$ are fresh variables.\\[1mm]
\textbf{(\dsolve)}&
$\displaystyle\frac{\langle\{s\labTriEq{x} t\}\cupdot A;S;T;\theta\rangle}{\langle A; \{s\labTriEq{x} t\}\cup S; T;\theta\rangle}$\\
&where $head(s)\neq head(t)$ and they are not related absorption symbols. \\[1mm]
\textbf{(\dexpalf)}&
$\displaystyle\frac{\langle\{\varepsilon_f\labTriEq{x} f(t_1,t_2)
\}\cupdot A;S; T;\theta\rangle}{\langle\{\varepsilon_f\labTriEq{y_1} t_1
\}\cup A; S;\{\star\labTriEq{y_2} t_2\}\cup T; \theta\{x\mapsto 
f(y_1,y_2)\}\rangle}$\\
&where $f$ is an absorption  symbol,  and $y_1,y_2$ are fresh variables.\\[1mm]
\textbf{(\dexpals)}&
$\displaystyle\frac{\langle\{\varepsilon_f\labTriEq{x} f(t_1,t_2)
\}\cupdot A;S;T;\theta\rangle}{\langle\{\varepsilon_f\labTriEq{y_2} t_2
\}\cup A; S;\{\star\labTriEq{y_1} t_1\}\cup T; \theta\{x\mapsto 
f(y_1,y_2)\}\rangle}$\\
 & where $f$ is an absorption symbol,  and $y_1,y_2$ are fresh variables.\\[1mm]
\textbf{(\dexparf)}&
$\displaystyle\frac{\langle\{f(s_1,s_2)\labTriEq{x} \varepsilon_f\}\cupdot A;S; T;\theta\rangle}{\langle\{s_1\labTriEq{y_1} \varepsilon_f\}\cup A; S;\{s_2\labTriEq{y_2} \star\}\cup T; \theta\{x\mapsto f(y_1,y_2)\}\rangle}$\\
&where $f$ is an absorption  symbol,  and $y_1,y_2$ are fresh variables.\\[1mm]
\textbf{(\dexpars)}&
$\displaystyle\frac{\langle\{f(s_1,s_2)\labTriEq{x} \varepsilon_f\}\cupdot A;S;T;\theta\rangle}{\langle\{s_2\labTriEq{y_2} \varepsilon_f\}\cup A; S;\{s_1\labTriEq{y_1} \star\}\cup T; \theta\{x\mapsto f(y_1,y_2)\}\rangle}$\\
 & where $f$ is an absorption symbol,  and $y_1,y_2$ are fresh variables.\\[1mm]
\textbf{(\dmer)} &
$\displaystyle\frac{\langle \emptyset;\{s\labTriEq{x} t, s\labTriEq{y} t\}\cup S; T;\theta\rangle}{\langle \emptyset;\{s\labTriEq{y} t\}\cup S;T; \theta\{x\mapsto y\}\rangle}$\\[3mm]
\hline
\end{tabular}
\end{table}

\begin{lemma}[Configuration preservation]
\label{lem:ConfPreserv}
Any quadruple $\mathcal{C'}=\langle A';S';T';\theta'\rangle$ derived from a configuration $\mathcal{C}=\langle A;S;T;\theta\rangle$ is a configuration.
\end{lemma}
\begin{proof}
    We analyze each possible rule application.

\begin{itemize}
    \item Rule \rdec\ includes fresh variables as new labels in the AUEs in $\mathcal{C}'$ and adds to the domain of the substitution the label of the  AUE eliminated from the AUEs in $\mathcal{C}$. This transformation does not violate condition (i) as all variables introduced into the range of the substitution are fresh. From this, we can deduce that condition (ii) also holds.
    
    \item Rule \rsolve\ moves an AUE from the unsolved AUEs in $\mathcal{C}$ to the store of $\mathcal{C}'$. Thus, conditions (i) and (ii) trivially hold in $\mathcal{C}'$ as the application of this rule does not introduce new AUEs.

    \item Rules \rexpalf, \rexpals, \rexpalf, and \rexpars\ include one fresh variable as a new label in the unsolved AUEs in $\mathcal{C}'$ and one fresh variable as a new label in the abstraction part, and add to the domain of the substitution the label of the  AUE eliminated from the AUEs in $\mathcal{C}$. Thus, condition (i) holds for $\mathcal{C}'$. Additionally, the new variables occur in the range of the substitution; thus, condition (ii) holds. 

    \item Rule \rmer\ eliminates a repeated AUE from the store of $\mathcal{C}$ adding its label to the domain of the substitution in $\mathcal{C}'$. Thus,  conditions (i) and (ii) hold. \qed
\end{itemize}
\end{proof}

An initial configuration is a configuration of the form $\langle A; \emptyset;\emptyset; \iota\rangle$, where the \textit{starting substitution} $\iota$ is given by a set of bindings $\{x_{st}\mapsto x\mid x\in \labels{A}\}$, each $x_{st}$ being a distinguished starting label. A \textit{normal configuration} is a configuration that is not reducible using the rules of \AU. A normal configuration resulting from a finite, exhaustive \AU\ derivation is called a \textit{final configuration}.   We denote the set of final configurations finitely derived from a configuration $\langle A; S; T ; \theta\rangle$ as $\AU(\langle A; S; T ; \theta\rangle)$.

\begin{theorem}[Termination]
\label{Thm:Ter}
$\AU$ cannot result in an infinite derivation.
Also, for a configuration $\mathcal{C}$, the set of final configurations $\AU(\mathcal{C})$ is computable in a finite number of steps. 
\end{theorem}
\begin{proof} 
After each rule application, except for  \rmer, the length of the unsolved component strictly decreases. Therefore, each possible derivation leaves a configuration with an empty unsolved component. The merge rule, \rmer, can only be applied when the unsolved component is empty, and it strictly decreases the number of AUEs in the store. Therefore, the procedure terminates.

Furthermore, since for any configuration, there are only finitely many ways to apply the \AU\ rules in Table~\ref{tab:generalizationrules} to a particular AUE in the unsolved component (finite branching),  we can use König's Lemma to conclude that the resulting set of final configurations is indeed finite and finitely computable.  \hfill $\Box$
\end{proof}



To give a brief description of the next steps we follow, assume an initial configuration $\langle A; \emptyset;\emptyset; \iota\rangle$ that leads to a final one $\langle \emptyset, S, T, \theta\rangle \in \AU(\langle A; \emptyset;\emptyset; \iota \rangle)$.  When we apply $\theta$ to initial labels in $A$, it results in terms describing the common structure of pairs of terms in the initial AUE. That is,  given the initial $s \labTriEq{x} t\in A$, we will prove that $x\theta \in \mathcal{G}_{\Abs}(s,t)$. Additionally, the final store $S$ is used to construct the substitutions of the $\Abs$-solutions, $\langle x\theta, \sigma,\rho\rangle$ such that $x\theta\sigma \approxOf{\Abs}  s$ and $x\theta\rho\approxOf{\Abs}  s$. The construction of these substitutions is formally described in Definition \ref{def:subsaue}.  
Finally, we use the final abstraction set $T$ to calculate generalizations derived from further expansions of terms using the absorption theory $\Abs$, showing that it produces $\mcsgabs(s,t)$.
\begin{definition}[Left and right substitutions]
\label{def:subsaue}Let $W$ be a finite valid set of AUEs. The left and right substitutions related to $W$ are defined as follows: $\sigma_{W} = \{y\mapsto s \;| \;s \labTriEq{y} t \in W\}$ and $ \rho_{W} = \{y\mapsto t \;|\; s\labTriEq{y} t \in W\}$.
\end{definition}
Since the algorithm is terminating and finitely branching, defining a finite set of computed generalizations is possible.

\begin{definition}[Computed solutions]
\label{def:comsol} Let  $\mathcal{D}=\langle A; S; T; \theta\rangle \Longrightarrow^{*} \langle \emptyset; S';T';\theta'\rangle$ be a derivation to a final configuration.
The computed generalization for the AUE $s\labTriEq{x}t\in A\cup S\cup T$ is defined as 
$\langle \{x\theta'\}_{x\in L}, \sigma_{\cal D}, \rho_{\cal D} \rangle$, where $\sigma_{\cal D} = \sigma_{S'\cup T'}$, $\rho_{\cal D} = \rho_{S'\cup T'}$, and $L=\labels{A}\cup \labels{S}\cup \labels{T}$. 

\end{definition}
\subsection{Abstraction Computation}
In this section, we construct the \textit{abstraction set} and substitutions from the store $S$ and abstractions $T$ computed by 
$\AU$. Let $\langle \{s\labTriEq{x} t\};\emptyset ; \emptyset; \iota\rangle$ be an initial configuration and 
$\langle \emptyset; S ; T; \theta \rangle\in \AU(\langle \{s\labTriEq{x} t\};\emptyset ; \emptyset; \iota\rangle )$. While $x\theta$ may be more specific than the 
syntactic generalization of $s$ and $t$, any use of the absorption theory while computing $x\theta$ is completely dependent on the presence 
of absorption symbols and constants within $s$ and $t$. It is not enough to capture all more specific generalizations than the syntactic 
generalization, as absorption allows for introducing additional structure beyond what is present in the initial AUE. For example, 
$\AU$ computes the generalization $f(x,y)$ for $\varepsilon_f$ and $ f(h(\varepsilon_f),h(h(\varepsilon_f)))$, yet $\Abs$ tells us that a 
more specific generalization, $f(x,h(x))$, also generalizes these terms. In more extreme cases, infinitely many more specific generalizations 
may exist.

\begin{definition}[Abstraction set]
\label{def:abstraction}
Let  $t$ be a term in \Abs-normal form, and $\sigma$ be a substitution whose range is in \Abs-normal form. The abstraction set of $t$ with respect to $\sigma$ is the set\vspace{-1mm}
\[\Abst(t,\sigma):=\{r \mid r\sigma \approxOf{\Abs} t,\  \mbox{$r$ is in an \Abs-normal form, and } \mathcal{V}(r)\subseteq \Dom(\sigma)\}.\]
\end{definition}

In words, $\Abst(t,\sigma)$ is the set of all those $\Abs$-generalizations of $t$ whose $\sigma$-instance is $t$, and they may contain only variables from $\Dom(\sigma)$. (Obviously, $t$ itself is in $\Abst(t,\sigma)$.) To obtain such an $r\in \Abst(t,\sigma)$, we can abstract some occurrences of some $x\sigma$'s in $t$ by $x$, where $x\in \Dom(\sigma)$; this is the origin of the term ``abstraction set''.
\begin{example} \label{ex:abstraction1}
Let $t=g(\varepsilon_f,f(h(a),b))$ and $\sigma=\{x\mapsto a, y\mapsto f(h(a),b),z\mapsto b \}$. Then the abstraction set of $t$ with respect to $\sigma$:

\noindent$\Abst(t,\sigma)=\{ t,\ g(\varepsilon_f,y),\ g(\varepsilon_f,f(h(x),b)),\ g(\varepsilon_f,f(h(a),z)), \ g(\varepsilon_f,f(h(x),z$ $))\}.$
Now, consider  $t=h(\varepsilon_f)$ and $\sigma=\{y\mapsto a, v\mapsto \varepsilon_f\}$. Then the abstraction set of $t$ with respect to $\sigma$:

\noindent$\Abst(t,\sigma)=\{h(\varepsilon_f),h(v)$ $,h(f(v,b)),h(f(a,v)),h(f(v,v)), h(f(v,f(y,a))), \dots\}.$ 

The latter $\Abst(t,\sigma)$ has infinitely many terms with variables from $\Dom(\sigma)$.
\end{example}

Given a configuration $\langle A;S;T;\theta\rangle$, the AUEs contained in $T$ are of the form $\star\triangleq_x t$ or $t\triangleq_x \star$ for some $t$. The labels occurring in $T$ also occur in the terms of the range of $\theta$. Here, we should interpret $\star$ as any term. Essentially, the abstraction substitution defined below extends $\theta$ by replacing the labels of $T$ with a generalization of the non-wildcard term of the associated AUE and some arbitrary term. While this is sufficient for constructing more specific generalizations, we consider restricting the variables occurring in the introduced terms. 

\begin{definition}[Abstraction substitutions]
\label{def:subsabs}
Let $\langle A;S;T;\theta\rangle$ be a configuration such that $T\neq \emptyset$. Then an \emph{abstraction substitution} of this configuration is any substitution $\tau$ such that
\begin{itemize}
    \item $\Dom(\tau) = \labels T$, and
    \item for each $y\in \Dom(\tau)$ we have $y\tau \in \Abst_{y}(T,S)$,
\end{itemize}
%
%
\noindent where the set of terms $\Abst_{y}(T,S)$ is defined as $\Abst_{y}(T,S)=\Abst(t,\rho_{S})$ if $\star\labTriEq{y} t \in T$ and $\Abst_{y}(T,S)=\Abst(s,\sigma_{S})$ if $s\labTriEq{y} \star \in T$. We denote the set of all possible abstraction substitutions of this configuration as $\Psi(T,S)$.
\end{definition}

\begin{corollary}
\label{cor:rangebnd}
Let $\langle A;S;T;\theta\rangle$ be a configuration such that $T\neq \emptyset$. Then for any $y\in labels(T)$ and $\tau\in \Psi(T,S)$, $\mathcal{V}(y\tau) \subseteq  labels(S)$.
\end{corollary}

We now provide a few examples illustrating the computation of final configurations using \AU, the construction of the abstraction sets, and substitutions associated with these final configurations. They highlight the complex structure of the computed sets of generalizations.

\begin{example}
\label{ex:procedure1}
Let $s=\varepsilon_f$ and $t= f(a,f(h(a),b))$. Applying \AU\  to the initial configuration $\langle \{s\labTriEq{x_{st}} t\};\emptyset ; \emptyset; \iota\rangle$  results in three different final configurations. 
%
%
\begin{align*}
\textbf{Configuration 1:} \qquad \qquad \qquad
\langle \{\varepsilon_f\labTriEq{x_{st}} f(a,f(h(a),b))\}; \emptyset;\emptyset;\iota\rangle &\dexpalf  \\[1mm]
          \langle \{\varepsilon_f\labTriEq{y} a\}; \emptyset;\{\star\labTriEq{z} f(h(a),b)\};\{x_{st}\mapsto f(y,z), x\mapsto f(y,z)\}\rangle& \dsolve\\[1mm]
          \langle \emptyset;\{\varepsilon_f\labTriEq{y} a\};\{\star\labTriEq{z} f(h(a),b)\};\{x_{st}\mapsto f(y,z), x\mapsto f(y,z)\}\rangle. &
    \end{align*}

Since $\Abst_z(f(h(a),b),\{y\mapsto a\})=\{f(h(a),b),f(h(y),b)\}$, the set of abstraction substitutions generated from the final abstraction and store ($T$ and $S$, respectively) is $\Psi(T,S)=\{\{z\mapsto f(h(a),b)\},\{z\mapsto f(h(y),b)\} \}$. Hence, the computed terms are $f(y,f(h(a),b))$ and $f(y,f(h(y),b)$, obtained by instantiating the initial label $x_{st}$ with the composition of the final substitution and the abstraction substitution. It is easy to see that the computed terms are indeed $\Abs$-generalizations of $\varepsilon_f$ and $f(a,f(h(a),b))$. The left and right substitutions are $\sigma =\{y\mapsto \varepsilon_f\}$, and $\rho =\{y\mapsto a\}$. 

From now on, we omit the starting label and denote the abstraction and store as $T$ and $S$, respectively.
%
%
\begin{align*}
\textbf{Configuration 2:} \qquad \qquad \qquad
\langle \{\varepsilon_f\labTriEq{x} f(a,f(h(a),b))\}; \emptyset;\emptyset;\iota\rangle&   \dexpals \\[1mm]
\langle \{\varepsilon_f\labTriEq{v} f(h(a),b)\}; \emptyset;\{\star\labTriEq{u} a\};\{x\mapsto f(u,v)\}\rangle &  \dexpalf \\[1mm]
\langle \{\varepsilon_f\labTriEq{y} h(a)\}; \emptyset;\{\star\labTriEq{u} a, \star\labTriEq{z} b\};\{x\mapsto f(u,f(y,z)),  v\mapsto f(y,z)\}\rangle &  \dsolve\\[1mm]
\langle\emptyset; \{\varepsilon_f\labTriEq{y} h(a)\};\{\star\labTriEq{u} a, \star\labTriEq{z} b\};\{x\mapsto f(u,f(y,z)),  v\mapsto f(y,z)\}\rangle.& 
\end{align*}

Since for the variables $u$ and $z$, $\Abst_u(T,S)=\Abst(a,\{y\mapsto h(a)\})=\{a\}$ and $\Abst_z(T,S)=\Abst(b,\{y\mapsto h(a)\})=\{b\}$, the abstraction substitution set $\Psi(T,S)$ is singleton, that is $\Psi(T,S)=\{\{u\mapsto a, z\mapsto b\}\}$ and the unique computed term in this configuration is  $f(a,f(y,b))$ (which is an $\Abs$-generalization of the input terms). The left and right substitutions are $\sigma =\{y\mapsto \varepsilon_f\}$, and $\rho =\{y\mapsto h(a)\}$.
%
\begin{align*}
\textbf{Configuration 3:} \qquad \qquad \qquad
\langle \{\varepsilon_f\labTriEq{x} f(a,f(h(a),b))\}; \emptyset;\emptyset;\iota\rangle&  \dexpals \\[1mm]
\langle \{\varepsilon_f\labTriEq{v} f(h(a),b)\}; \emptyset;\{\star\labTriEq{u} a\};\{x\mapsto f(u,v)\}\rangle&  \dexpals \\[1mm]
\langle \{\varepsilon_f\labTriEq{z} b\}; \emptyset;\{\star\labTriEq{u} a, \star\labTriEq{y} h(a)\};\{x\mapsto f(u,f(y,z)), v\mapsto f(y,z)\}\rangle& \dsolve\\[1mm]
\langle\emptyset; \{\varepsilon_f\labTriEq{z} b\};\{\star\labTriEq{u} a, \star\labTriEq{y} h(a)\};\{x\mapsto f(u,f(y,z)), v\mapsto f(y,z)\}\rangle. & 
\end{align*}

Since $\Abst_u(T,S)=\Abst(a,\{z\mapsto b\})=\{a\}$ and $\Abst_y(T,S)=\Abst(h(a),\{z\mapsto b\})=\{h(a)\}$, the abstraction substitution set is $\Psi(T,S)=\{\{u\mapsto a, y\mapsto h(a)\}\}$ and the unique answer on this branch is $f(a,f(h(a),z))$ (which is an $\Abs$-generalization of the $s$ and $t$). The left and right substitutions are $\sigma =\{y\mapsto \varepsilon_f\}$ and $\rho =\{y\mapsto b\}$.

Hence, for terms  $\varepsilon_f$ and $f(a,f(h(a),b))$, after abstracting the {\AU} computed solutions, we get the following set of $\Abs$-generalizations:
\[ \{f(y,f(h(a),b)),\ f(y,f(h(y),b), \ f(a,f(y,b)),\ f(a,f(h(a),z))\}.\]
\end{example}
The next example shows the application of $\AU$ to a more complex AUE, resulting in an infinite set of generalizations. 
\begin{example}
\label{ex:procedure2} Applying \AU~to $s\triangleq t$ where $s=g(\varepsilon_f,f(a,h(\varepsilon_f)))$ and $t= g(f(h(\varepsilon_f),a),\varepsilon_f)$ results in the following four configurations.
%
\begin{align*}
\textbf{Configuration 1:} \qquad 
\langle 
\{g(\varepsilon_f,f(a,h(\varepsilon_f)))\triangleq g(f(h(\varepsilon_f),a),\varepsilon_f)\}; \emptyset;\emptyset;\iota\rangle &  \ddec  \\[1mm]
\langle \{\varepsilon_f\labTriEq{w_1} f(h(\varepsilon_f),a), f(a,h(\varepsilon_f))\labTriEq{w_2} \varepsilon_f\}; \emptyset;\emptyset;\{x\mapsto g(w_1,w_2)\}\rangle  &\dexpalf \\[1mm]
\langle \{\varepsilon_f\labTriEq{u_1} h(\varepsilon_f), f(a,h(\varepsilon_f))\labTriEq{w_2} \varepsilon_f\}; \emptyset;\{\star\labTriEq{v_1} a\};&\\
\{x\mapsto g(f(u_1,v_1),w_2), w_1\mapsto f(u_1,v_1)\}\rangle & \dexparf \\[1mm]
\langle \{\varepsilon_f\labTriEq{u_1} h(\varepsilon_f), a\labTriEq{u_2} \varepsilon_f\}; \emptyset;\{\star\labTriEq{v_1} a,h(\varepsilon_f)\labTriEq{v_2} \star\};&\\
\{x\mapsto g(f(u_1,v_1),f(u_2,v_2)),  w_1\mapsto f(u_1,v_1), w_2\mapsto f(u_2,v_2)\}\rangle& \ddsolve{\times 2} \\[1mm]
\langle \emptyset; \{\varepsilon_f\labTriEq{u_1} h(\varepsilon_f), a\labTriEq{u_2} \varepsilon_f\};\{\star\labTriEq{v_1} a,h(\varepsilon_f)\labTriEq{v_2} \star\};&\\
\{x\mapsto g(f(u_1,v_1),f(u_2,v_2)),  w_1\mapsto f(u_1,v_1), w_2\mapsto f(u_2,v_2)\}\rangle
\end{align*}

Note, for the variable $v_1$, $\Abst_{v_1}(T,S)=\Abst(a,\{u_1\mapsto h(\varepsilon_f),u_2\mapsto \varepsilon_f\})=\{a\}$, and for the variable $v_2$,  $\Abst_{v_2}(T,S)=$ $\Abst(h(\varepsilon_f),\{u_1\mapsto \varepsilon_f,u_2\mapsto a\})$ is an infinite set, including $\{h(\varepsilon_f), h(u_1), h(f(u_1,a)),$ $  h(f(a, u_1)),  h(f(u_1,u_2)),\dots\}$.

Here, the set of abstraction substitutions $\Psi(T,S)$ is an infinite set including $\{\{v_1\mapsto a, v_2\mapsto h(\varepsilon_f)\}, \{v_1\mapsto a, v_2\mapsto h(u_1)\}, \{v_1\mapsto a, v_2\mapsto h(f(u_1,a))\}, \ldots\}$.

Some of the terms computed on this branch are $g(f(u_1,a)),f(u_2,\allowbreak h(\varepsilon_f)))$, $g(f(u_1,a)),f(u_2,h(u_1)))$, and $ g(f(u_1,a)),f(u_2,h(f(u_1,a))))$. They all $\Abs$-generalize the given terms $s$ and $t$.
%
\begin{align*}
\textbf{Configuration 2:} \qquad 
\langle \{g(\varepsilon_f,f(a,h(\varepsilon_f)))\triangleq g(f(h(\varepsilon_f),a),\varepsilon_f)\}; \emptyset;\emptyset;\iota\rangle& \ddec \\[1mm]
\langle \{\varepsilon_f\labTriEq{w_1} f(h(\varepsilon_f),a), f(a,h(\varepsilon_f))\labTriEq{w_2} \varepsilon_f\}; \emptyset;\emptyset;\{x\mapsto g(w_1,w_2)\}\rangle &  \dexpalf \\[1mm]
\langle \{\varepsilon_f\labTriEq{u_1} h(\varepsilon_f), f(a,h(\varepsilon_f))\labTriEq{w_2} \varepsilon_f\}; \emptyset;\{\star\labTriEq{v_1} a\};&\\
\{x\mapsto g(f(u_1,v_1),w_2),w_1\mapsto f(u_1,v_1)\}\rangle&  \dexpars\\[1mm]
\langle \{\varepsilon_f\labTriEq{u_1} h(\varepsilon_f), h(\varepsilon_f)\labTriEq{v_2} \varepsilon_f\}; \emptyset;\{\star\labTriEq{v_1} a,a\labTriEq{u_2} \star\};&\\
\{x\mapsto g(f(u_1,v_1),f(u_2,v_2)), w_1\mapsto f(u_1,v_1), w_2\mapsto f(u_2,v_2)\}\rangle&  \ddsolve{\times 2}\\[1mm]
\langle \emptyset; \{\varepsilon_f\labTriEq{u_1} h(\varepsilon_f), h(\varepsilon_f)\labTriEq{v_2} \varepsilon_f\};\{\star\labTriEq{v_1} a,a\labTriEq{u_2} \star\};&\\
\{x\mapsto g(f(u_1,v_1),f(u_2,v_2)), w_1\mapsto f(u_1,v_1), w_2\mapsto f(u_2,v_2)\}\rangle& 
\end{align*}

Then for the variable $v_1$, $\Abst_{v_1}(T,S)=\Abst(a,\{u_1\mapsto h(\varepsilon_f),v_2\mapsto \varepsilon_f\})=\{a\}$ and for $u_2$, 
$\Abst_{u_2}(T,S)=\Abst(a,\{u_1\mapsto \varepsilon_f,v_2\mapsto h(\varepsilon_f)\})=\{a\}$, i.e. on this branch we get $g(f(u_1,a)),f(a,v_2))$ (yet another generalization of $s$ and $t$).
%
\begin{align*}
\textbf{Configuration 3:} \qquad 
\langle \{g(\varepsilon_f,f(a,h(\varepsilon_f)))\triangleq g(f(h(\varepsilon_f),a),\varepsilon_f)\}; \emptyset;\emptyset;\iota\rangle&  \ddec \\[1mm]
\langle \{\varepsilon_f\labTriEq{w_1} f(h(\varepsilon_f),a), f(a,h(\varepsilon_f))\labTriEq{w_2} \varepsilon_f\}; \emptyset;\emptyset;\{x\mapsto g(w_1,w_2)\}\rangle & \dexpals \\[1mm]
\langle \{\varepsilon_f\labTriEq{v_1} a, f(a,h(\varepsilon_f))\labTriEq{w_2} \varepsilon_f\}; \emptyset;\{\star\labTriEq{u_1} h(\varepsilon_f)\};&\\
\{x\mapsto g(f(u_1,v_1),w_2), w_1\mapsto f(u_1,v_1)\}\rangle& \dexparf\\[1mm]
\langle \{\varepsilon_f\labTriEq{v_1} a, a\labTriEq{u_2} \varepsilon_f\}; \emptyset;\{\star\labTriEq{u_1} h(\varepsilon_f),h(\varepsilon_f)\labTriEq{v_2} \star\};&\\
\{x\mapsto g(f(u_1,v_1),f(u_2,v_2)),  w_1\mapsto f(u_1,v_1),  w_2\mapsto f(u_2,v_2)\}\rangle&  \ddsolve{\times 2}\\[1mm]
\langle \emptyset; \{\varepsilon_f\labTriEq{v_1} a, a\labTriEq{u_2} \varepsilon_f\};\{\star\labTriEq{u_1} h(\varepsilon_f),h(\varepsilon_f)\labTriEq{v_2} \star\};&\\
\{x\mapsto g(f(u_1,v_1),f(u_2,v_2)), w_1\mapsto f(u_1,v_1),  w_2\mapsto f(u_2,v_2)\}\rangle&
\end{align*}

\noindent Then $\Abst_{u_1}(T,S)=\Abst(h(\varepsilon_f),\{v_1\mapsto a,u_2\mapsto \varepsilon_f\})=\{h(\varepsilon_f),h(u_2),h(f(u_2,a)),\allowbreak \dots\}$ and
$\Abst_{v_2}(T,S)=\Abst(h(\varepsilon_f),\{v_1\mapsto \varepsilon_f,u_2\mapsto a\})=\{h(\varepsilon_f), h(v_1), h(f(v_1,a)),\allowbreak \dots \}$.
The set $\Psi(T,S)$ is infinite $\{\{u_1\mapsto h(\varepsilon_f), v_2\mapsto h(\varepsilon_f)\},$ $ \{u_1\mapsto h(\varepsilon_f), v_2\mapsto h(v_1)\}, \{v_1\mapsto a, v_2\mapsto h(f(u_1,a))\}, \ldots\}$. Some of the terms computed on this branch are $g(f(h(\varepsilon_f),v_1)),f(u_2,h(\varepsilon_f)))$, $g(f(h(u_2),v_1)),f(u_2, h(v_1)))$. They all are $\Abs$-generalizations of the given terms.
\vspace{-1.5mm}
%
\begin{align*}
\textbf{Configuration 4:} \qquad 
\langle \{g(\varepsilon_f,f(a,h(\varepsilon_f)))\triangleq g(f(h(\varepsilon_f),a),\varepsilon_f)\}; \emptyset;\emptyset;\iota\rangle&  \ddec \\[1mm]
\langle \{\varepsilon_f\labTriEq{w_1} f(h(\varepsilon_f),a), f(a,h(\varepsilon_f))\labTriEq{w_2} \varepsilon_f\}; \emptyset;\emptyset;\{x\mapsto g(w_1,w_2)\}\rangle & \dexpals \\[1mm]
\langle \{\varepsilon_f\labTriEq{v_1} a, f(a,h(\varepsilon_f))\labTriEq{w_2} \varepsilon_f\}; \emptyset;\{\star\labTriEq{u_1} h(\varepsilon_f)\};&\\
\{x\mapsto g(f(u_1,v_1),w_2), w_1\mapsto f(u_1,v_1)\}\rangle& \dexpars\\[1mm]
\langle \{\varepsilon_f\labTriEq{v_1} a, h(\varepsilon_f)\labTriEq{v_2} \varepsilon_f\}; \emptyset;\{\star\labTriEq{u_1} h(\varepsilon_f), a\labTriEq{u_2} \star\};&\\
\{x\mapsto g(f(u_1,v_1),f(u_2,v_2)),  w_1\mapsto f(u_1,v_1),  w_2\mapsto f(u_2,v_2)\}\rangle&  \ddsolve{\times 2}\\[1mm]
\langle \emptyset; \{\varepsilon_f\labTriEq{v_1} a, h(\varepsilon_f)\labTriEq{v_2} \varepsilon_f\};\{\star\labTriEq{u_1} h(\varepsilon_f),a\labTriEq{u_2} \star\};&\\
\{x\mapsto g(f(u_1,v_1),f(u_2,v_2)), w_1\mapsto f(u_1,v_1),  w_2\mapsto f(u_2,v_2)\}\rangle&
\end{align*}

Then $\Abst_{u_1}(T,S)=\Abst(h(\varepsilon_f),\{v_1\mapsto a,v_2\mapsto \varepsilon_f\})=\{h(\varepsilon_f),h(v_2),h(f(v_2,a)),\allowbreak \ldots\}$, $\Abst_{u_2}(T,S)=\Abst(a,\{v_1\mapsto \varepsilon_f,v_2\mapsto h(\varepsilon_f)\})=\{a\}$, and some of the terms computed of this branch are  $g(f(h(\varepsilon_f),v_1)),f(a,v_2))$,\, $g(f(h(v_2),v_1)),f(a,v_2))$, and $ g(f(h(f(v_2,\allowbreak a),v_1)),f(a,v_2))$. They all $\Abs$-generalize $s$ and $t$.
\end{example}





\section{Soundness and Completeness}
\label{sec:SoundnessAndCompleteness}

Preserving the stated properties of configurations (Definition \ref{def:configuration}) is essential to both the soundness and completeness proofs as these properties enforce consistency with respect to the use of the labels.
\begin{theorem}[Soundness]
\label{Thm:sound}
Let  $\mathcal{D}=\langle A_0; S_0; T_0; \theta_0\rangle \Longrightarrow^{*} \langle \emptyset; S_n;T_n;\theta_n\rangle$ be a derivation to a final configuration and $\langle \{x\theta_n\}_{x\in L}, \sigma_{\cal D}, \rho_{\cal D}  \rangle$ be the corresponding computed solution (see Definition~\ref{def:comsol}). Then for all $s\labTriEq{x} t\in A_0\cup S_0\cup T_0$, we have $x \theta_n\in \mathcal{G}_{\Abs}(s,t)$ with  $x \theta_n \sigma_{\mathcal{D}} \approxOf{\Abs} s$ and  $x \theta_n\rho_{\mathcal{D}}\approxOf{\Abs} t$.
\end{theorem}

\begin{proof} 
We proceed by induction over the derivation length.\vspace{0.2cm}

\noindent{\bf Basecase}. 
If the derivation has length 0, then it starts with a final configuration, and we only need to consider AUEs in $S_0$ and $T_0$. By the properties of configurations (Definition~ \ref{def:configuration})  we get $x\theta_0=x$ for all $s\labTriEq{x} t\in S_0 \cup T_0$. Also, $x\theta_0\sigma_{\cal D}=s$ and $x\theta_0\rho_{\cal D}=t$ hold. \vspace{0.2cm}

\noindent{\bf Stepcase.}  Now consider our derivation having the following form: \vspace{-1mm}
\begin{equation}\label{eq:inducdev}
\langle A_0;S_0;T_0;\theta_0 \rangle\Longrightarrow\langle A_1; S_1; T_1; \theta_1\rangle \Longrightarrow^n \langle \emptyset; S_{n+1};T_{n+1};\theta_{n+1}\rangle \end{equation} 

We assume for the induction hypothesis (IH) that for   $\langle A_1; S_1; T_1; \theta_1\rangle \Longrightarrow^n \langle \emptyset; S_{n+1};T_{n+1};\theta_{n+1}\rangle$ the theorem holds and show it for derivations of the form presented in (\ref{eq:inducdev}). We continue the proof considering the various options for the transition from $\langle A_0;S_0;T_0;\theta_0 \rangle$ to $\langle A_1;S_1;T_1;\theta_1\rangle$.

\begin{enumerate}
\item {\bf \rdec}. Assume that the derivation is of the form: 
\vspace{-1.5mm}
\begin{align*}
\langle \{f(s_1,\dots, s_m)\labTriEq{y} f(t_1,\dots, t_m)\}\cupdot A';S_0;T_0;\theta_0 \rangle\ddec\\
\langle \{s_1\labTriEq{x_1} t_1,\dots,s_m\labTriEq{x_m} t_m\}\cupdot A'; S_1; T_1;\theta_1\rangle\Longrightarrow^n \langle \emptyset; S_{n+1};T_{n+1};\theta_{n+1}\rangle&
\end{align*}
where $\theta_1 = \theta_0\{y\mapsto f(x_1,\dots,x_m)\}$.  By the IH, we have $x_i\theta_{n+1}\in \mathcal{G}_{\Abs}(s_i,t_i)$ for each $1\leq i\leq m$, with $x_i\theta_{n+1}\sigma_{\mathcal{D}}\approxOf{\Abs} s_i$ and $x_i\theta_{n+1}\rho_{\mathcal{D}}\approxOf{\Abs} t_i$. Hence, $y\theta_{n+1}\sigma_{\mathcal{D}}=f(x_1\theta_{n+1}\sigma_{\mathcal{D}},\dots,x_m\theta_{n+1}\sigma_{\mathcal{D}})\approxOf{\Abs} f(s_1,\dots, s_m)$ and $y\theta_{n+1}\rho_{\mathcal{D}}=f(x_1\theta_{n+1}\rho_{\mathcal{D}},\dots,x_m\theta_{n+1}\rho_{\mathcal{D}})\approxOf{\Abs} f(t_1,\dots, t_m)$.

\item {\bf \rsolve}. Assume that the derivation is of the form: \vspace{-1.5mm}
\[\langle\{s\labTriEq{y} t\}\cupdot A';S_0;T_0;\theta_0\rangle \dsolve
\langle\{A';S_1  ;T_0;\theta_0\}\rangle\Longrightarrow^{n} \langle \emptyset;S_{n+1};T_{n+1};\theta_{n+1}\rangle,\]
where $S_1 =\{s\labTriEq{y} t\} \cup S_0$. By IH, $\theta_{n+1}$ generalizes all the AUEs with labels in $S_1$. Thus, $y\theta_{n+1}\in \mathcal{G}_{\Abs}(s,t)$ with $y\theta_{n+1}\sigma_{\mathcal{D}}\approxOf{\Abs} s$ and $y\theta_{n+1}\rho_{\mathcal{D}}\approxOf{\Abs} t$.

\item {\bf \rexpalf}. Assume that the derivation is of the form:\vspace{-1.5mm}
\begin{align*}
  \langle \{\varepsilon_f\labTriEq{y} f(s, t)\}\cupdot A';S_0;T_0;\theta_0 \rangle\dexpalf\\
\langle \{\varepsilon_f\labTriEq{x_1} s\}\cup A'; S_1; T_1; \theta_0^1\}\rangle\Longrightarrow^n \langle \emptyset; S_{n+1};T_{n+1};\theta_{n+1}\rangle 
\end{align*}
where $T_1 =\{\star\labTriEq{x_2} t\}\cup T_0$ and $\theta_1 =\theta_0\{y\mapsto f(x_1,x_2)$\}. By the IH, all the AUEs in $S_1\cup T_1$ are generalized by the substitution  $\theta_{n+1}$. This implies that $x_1\theta_{n+1}\in \mathcal{G}_{\Abs}(\varepsilon_f,s)$, and $x_2\theta_{n+1}\in \mathcal{G}_{\Abs}(\star,t)$ with the substitutions $\sigma_{\mathcal{D}}$ and $\rho_{\mathcal{D}}$, respectively. Additionally, $y\theta_{n+1}\sigma_{\mathcal{D}}=f(x_1\theta_{n+1}\sigma_{\mathcal{D}},x_2\theta_{n+1}\sigma_{\mathcal{D}})\approxOf{\Abs} f(\varepsilon_f,\star)\approxOf{\Abs} \varepsilon_f$ and  $y\theta_{n+1}\rho_{\mathcal{D}}=f(x_1\theta_{n+1}\rho_{\mathcal{D}},x_2\theta_{n+1}\rho_{\mathcal{D}})\approxOf{\Abs} f(s,t)$. Hence,  $y\theta_{n+1}\in \mathcal{G}_{\Abs}(\varepsilon_f,f(s,t))$.

\item The analysis of other expansion rules 
is similar to the analysis of {\bf \rexpalf}.

\item {\bf \rmer} Assume that the derivation is of the form: \vspace{-1.5mm}
\begin{align*}
   \langle\emptyset;\{s\labTriEq{y} t, s\labTriEq{z} t\}\cup S';T_0;\theta_0\rangle \dmer\\
\langle\{\emptyset;\{s\labTriEq{z} t\} \cup S'  ;T_0;\theta_0\{y\mapsto z\}\}\rangle\Longrightarrow^{n} \langle \emptyset;S_{n+1};T_{n+1};\theta_{n+1}\rangle. 
\end{align*}

Notice that $\theta_{1}=\theta_0\{y\mapsto z\}$, where $z$ is the label of the AUE  $\{s\labTriEq{z} t\}\in S_0$. By IH, $z\theta_{n+1}$ is a generalization of $s$ and $t$. Then, $y\theta_{n+1}=y\{y\mapsto z\}\theta_{n+1}=z\theta_{n+1}$ is a generalization of $s$ and $t$ with substitutions $\sigma_{\mathcal{D}}$ and $\rho_{\mathcal{D}}$, respectively.   \qed
\end{enumerate}
\end{proof}

While the soundness theorem covers the construction of generalizations of AUEs present in a given configuration, it does not consider the abstraction set or the construction of more specific generalizations when generalizing over an absorption theory. The abstraction set allows us to consider generalizations between a given term and an arbitrary term.  
\begin{lemma}
\label{wildlemma}
Let  $\langle A_0; S_0; T_0; \theta_0\rangle \Longrightarrow^{*} \langle \emptyset; S_n;T_n;\theta_n\rangle$ be a derivation. Then for all $\star\labTriEq{u} t\in T_n$ (resp. for all $s \labTriEq{u} \star \in T_n$) and $\tau\in\Psi(T_n,S_n)$, there exists a term $r$ such that $u\tau\in\mathcal{G}_{\Abs}(r,t)$ (resp. $u\tau\in\mathcal{G}_{\Abs}(r,s)$).
\end{lemma}

\begin{proof}
The proof is by induction over the $len(u\tau)$. We prove it for $\star\labTriEq{u} t\in T_n$. The other case is analogous.
\begin{itemize}
    \item If $len(u\tau)=1$ then $u\tau=x$ or $u\tau=a$. If $u\tau=x$, then $r$ could be any term, and the lemma holds. If $u\tau=a$, then $t=a$ and, from the definition of the abstraction set,  $r=t$ and the lemma holds again.
    \item If $len(u\tau)=n$, then we can assume $u\tau=f(s_1,\dots,s_n)$, from the definition of the abstraction set we have $t=f(t_1,\dots,t_n)$, and by the IH there exists $r_i$ such that $s_i\in \mathcal{G}_{\Abs}(t_i,r_i)$ for $1\leq i\leq n$. Therefore, we can take $r=f(r_1,\dots,r_n)$. \qed
\end{itemize}
\end{proof}

This result intuitively means the following: observe that $u\tau\in \Abst_u(T_n,S_n)$ implying that $u\tau\in \Abst(t,\sigma_\mathcal{D})$. From this observation, we can deduce that $u\tau\sigma_\mathcal{D}\approxOf{\Abs} t$. Thus, for every AUE in the set $T_n$, the wild card can be interpreted as $r$ and $u\tau\rho_\mathcal{D}\approxOf{\Abs} r$. It leads to the following result: 

 \begin{theorem}
 \label{thm:snd_tau}
Let  $\langle A_0; S_0; T_0; \theta_0\rangle \Longrightarrow^{*} \langle \emptyset; S_n;T_n;\theta_n\rangle$ be a derivation to a final configuration and $s\labTriEq{x} t\in A_0\cup S_0$. Then $x\theta_n\tau \in \mathcal{G}_{\Abs}(s,t)$, where $\tau\in\Psi(T_n,S_n)$.
 \end{theorem}

\begin{proof}
From Theorem~\ref{Thm:sound}, $x\theta_n\in \mathcal{G}_{\Abs}(s,t)$. 
Furthermore, every $u\in \labels{T_n}$ is unique, only occurs once in $x\theta_n$, and $u\theta_n\tau=u\tau$. Considering these facts together with  Lemma~\ref{wildlemma} and $u$ being an $\Abs$-generalization of the respective subterms in $s$ and $t$, we deduce that $x\theta_n \tau\in \mathcal{G}_{\Abs}(s,t)$.
\qed
\end{proof} 

\noindent We are now ready to prove the completeness of $\AU$.
\begin{theorem}[Completeness]
\label{Thm:Completeness1}
Let  $r\in \mathcal{G}_{\Abs}(t_1,t_2)$. 
Then for every configuration $\langle A; S; T; \theta\rangle$ such that  $t_1\labTriEq{x} t_2 \in A$ for some $x$ there exist a final configuration $\langle \emptyset; S';T';\theta'\rangle\in \AU(\langle A; S;T;\theta\rangle)$ and $\tau\in\Psi(T',S')$ such that $r\genOf{\Abs}x\theta'\tau$.
\end{theorem}

\begin{proof}
The proof is by structural induction over $r$. \vspace{-1mm}
\begin{center} \textbf{Basecase}\vspace{-2mm}\end{center}
  \begin{enumerate}
  \item Let $r$ be a variable. Then, we must consider the following three cases: 
    \begin{enumerate}
        \item If $head(t_1)=head(t_2)$, then from a configuration $\langle A; S; T; \theta\rangle$ such that  $t_1\labTriEq{x} t_2 \in A$, we can reach the configuration   $\langle A'; S; T; \theta'\rangle$ by applying the decomposition rule such that $head(x\theta')=head(t_1)=head(t_2)$. Thus, for any final configuration $\langle \emptyset; S''; T''; \theta''\rangle\in \AU(\langle A'; S; T; \theta'\rangle)$, $r\genOf{\Abs}x\theta''$ as $\theta''$ can only be more specific than $\theta'$. 
        \item If, w.l.o.g, $t_1 = \varepsilon_f$ and $t_2=f(s_1,s_2)$, then from a configuration $\langle A; S; T; \theta\rangle$ such that  $t_1\labTriEq{x} t_2 \in A$, we can reach a configuration   $\langle A'; S'; T'; \theta'\rangle$ using the expansion rule \textit{ExpLA1} such that $head(x\theta')=head(t_2)$. Thus, for any final configuration $\langle \emptyset; S''; T''; \theta''\rangle\in \AU(\langle A'; S'; T'; \theta'\rangle)$, $r\genOf{\Abs}x\theta''$ as $\theta''$ can only be more specific than $\theta'$.
        \item Otherwise, if $head(t_1)\not =head(t_2)$, then from a configuration $\langle A; S; T; \theta\rangle$ such that  $t_1\labTriEq{x} t_2 \in A$, we can reach a configuration $\langle A'; S'; T; \theta\rangle$ using the \textit{Solve} rule where  $t_1\labTriEq{x} t_2 \in S$. Thus, for any final configuration $\langle \emptyset; S''; T''; \theta''\rangle\in \AU(\langle A'; S'; T; \theta\rangle)$, we get $r\approxOf{\Abs} x\theta''$.
    \end{enumerate}
    In all three cases $r\genOf{\Abs} x\theta''$ and by Theorem~\ref{thm:snd_tau} we get $r\genOf{\Abs} x\theta'\tau$.
  \item Let $r$ be a constant. Then $t_1=t_2=r$  and from a configuration $\langle A; S; T; \theta\rangle$ where  $t_1\labTriEq{x} t_2 \in A$, we can reach a configuration   $\langle A'; S; T; \theta'\rangle$ using the decomposition rule such that $x\theta'=t_1=t_2=r$. Thus,  for any final configuration $\langle \emptyset; S''; T''; \theta''\rangle\in \AU(\langle A'; S'; T'; \theta'\rangle)$, $r\genOf{\Abs} x\theta''\tau$ trivially follows.\vspace{-1mm}
 \end{enumerate}
 \begin{center} \textbf{Stepcase}\vspace{-2mm}\end{center}
 \begin{enumerate}
     \item $r=g(r_1,\ldots,r_n)$, $t_1=g(t'_1,\ldots,t'_n)$, and $t_2=g(t''_1,\ldots,t''_n)$; This implies that 
     $r_i$ is a generalization of $t'_i \labTriEq{y_i} t''_i$ for $1\leq i \leq n$. From $\langle 
     A; S; T; \theta\rangle$ we can reach a state $\langle A'; S'; T'; \theta'\rangle$, using the decomposition 
     rule, such that  $t'_i \labTriEq{y_i} t''_i\in A'$.

\hspace{1em} Note that there may exist  $1\leq i < j \leq n$ such that $\mathcal{V}(r_i)\cap \mathcal{V}(r_j)\not = \emptyset$. Let $R\subseteq \mathcal{V}(r)$ such that for $z\in R$ there exist $1\leq i < j \leq n$ such that $z\in \mathcal{V}(r_i)\cap \mathcal{V}(r_j)$. For any $z\in R$, there are two 
cases to consider: 
\begin{enumerate}[label=(\roman*)] 
    \item  There does not exist a position $p\in \mathit{pos}(t_1)\cap\mathit{pos}(t_2)$ such that the variable $z$ generalizes $s^*\triangleq t^*$ where $s^*=t_1\vert_p$ and  $t^* = t_2\vert_p$. In other words, $z$ 
generalizes terms which are absorbed during \Abs-normalization of $r\sigma$ and $r\rho$, where $r\sigma
\approxOf{\Abs} t_1$ and $r\rho\approxOf{\Abs} t_2$; this implies that replacing occurrences of $z$ by $\varepsilon_f$ (for the appropriate absorption symbol $f$) within $r$ results in a more specific generalization $r'$. For the remainder of this proof, we can consider $r$ to be the generalization resulting from replacing all such variables in $R$ by the appropriate absorption constant $\varepsilon_f$.

\item     There exists a position $p\in \mathit{pos}(t_1)\cap\mathit{pos}(t_2)$ such that $z$ generalizes $s^*\triangleq t^*$ where $s^*=t_1\vert_p$ and  $t^* = t_2\vert_p$. Notice that $z$ is structurally smaller than $r$ and thus, by the IH, there exists a final configuration $\langle \emptyset; S^*; T^*; \theta^*\rangle\in \AU(\langle \{s^*\triangleq t^* \}; \emptyset; \emptyset; \iota\rangle)$ and $\tau^*\in\Psi(T^*, S^*)$ such that $z\leq x'\theta^*\tau^*$. We will use $\theta^*\tau^*$ to align the generalizations resulting from the IH.
   \end{enumerate}

    \hspace{1em}By the induction hypothesis, there exists a 
     final configuration $\langle \emptyset; S''; T'';$ $ \theta''\rangle\in \AU(\langle A'; S'; T'; \theta'\rangle)$ 
     and $\tau_i\in\Psi(T'',S'')$ such that $r_i\genOf{\Abs} y_i\theta''\tau_i$ where $1\leq i\leq n$. Note,
     we can choose the same configuration $\langle \emptyset; S''; T''; \theta''\rangle$ for all AUEs  
     $t'_i\triangleq t''_i$ labeled by $y_i$.  Furthermore, we can choose $\langle \emptyset; S''; T''; \theta''\rangle$ such that $S^*\subseteq S''$  and $T^*\subseteq T''$ modulo label renaming as $s^*$ and $t^*$ are subterms of $t_1$ and $t_2$, respectively, modulo absorption symbol introduction. Now, we define $\gamma_i$ as the substitution such that $r_i\gamma_i\approxOf{\Abs}y_i\theta''\tau_i$. By the above construction, we can safely assume for all $z\in \mathcal{V}(r_1)\cap\mathcal{V}(r_2)$ such that $z$ has not been replaced by an absorption constant, that $z\gamma_i\approxOf{\Abs} z\theta^*\tau^*$ as there exist AUEs corresponding to $S^*$ and $T^*$ in $S''$ and $T''$, respectively.

 \hspace{1em}Now let $\mu$ be a substitution and $r_i'$ ($1\leq i \leq n$) be terms  such that for all $1\leq i\leq n$,  $r_i= 
    r_i'\mu$ and $g(r_1',\ldots,r_n') \preceq_{\tiny{\Abs}} 
          g(y_1\theta'',\ldots, y_n\theta'')$. If $\mu$ is the identity substitution, then we are done. Otherwise, we can use $\mu$ to construct a $\tau\in\Psi(T'',S'')$. Additionally, we need to consider the  $\tau_i
          \in\Psi(T'',S'')$ derived above for each $r_i$, where $1\leq i\leq n$, and the corresponding substitutions $\gamma_i$. Thus, $r_i'\mu \genOf{\Abs} 
          y_i\theta''\tau_i$ and  $r_i'\mu\gamma_i\approxOf{\Abs}
          y_i\theta''\tau_i$. 
          
          \hspace{1em}Now let $\mu_i^1$ and $\mu_i^2$ be substitutions such that  $\mu\gamma_i = (\mu_i^1\mu_i^2)\vert_{\Dom(\mu\gamma_i)}$ and $r'_i\mu_i^1\approxOf{\Abs}y_i\theta''$. This is possible given our assumption that $g(r_1',\ldots,r_n') \preceq_{\tiny{\Abs}} 
          g(y_1\theta'',\cdots, y_n\theta'')$. Note that $r'_i\mu_i^1\approxOf{\Abs} y_i\theta''$ implies that for every $x\in  \Dom(\mu_i^2)$ there exists a $z\in \Dom(\tau_i)$ such that $z\tau_i\approx_{\tiny abs} x\mu_i^2$. 

          \hspace{1em}We now construct $\tau\in \Psi(T'',S'')$ using the $\mu_i^2$, that is for all $1\leq  j\leq n$ and $x\in \Dom(\mu_j^2)$ there exists a $z\in \Dom(\tau)$ such that  $z\tau\approxOf{\Abs} x\mu_j^2$. It now follows that $r_i \genOf{\Abs}  y_i\theta''\tau$ holds for all  $1\leq 
          i\leq n$ and thus we have shown that $g(r_1,\ldots,r_n)  
          \genOf{\Abs}  g(y_1,\cdots, y_n)\theta''\tau$.

      \item $r=f(r_1,r_2)$, where $f$ is an absorption symbol and, w.l.o.g,  $t_1=\varepsilon_f$ and $t_2=f(s_1,s_2)$. Then 
      from $\langle A; S; T; \theta\rangle$ we can derive a configuration  $\langle A'; S'; T'; \theta'\rangle$ using the \textit{ExpLA1} rule such that $\star\labTriEq{y_2} s_2\in T'$ and $\varepsilon_f\labTriEq{y_1} s_1\in A'$. Now let $\langle \emptyset; S''; T''; \theta''\rangle\in \AU(\langle A'; S'; T'; \theta'\rangle)$ be a final configuration. 

\hspace{1em}By the induction hypothesis we know that $r_1\preceq_{\tiny abs} y_1\theta''\tau_1$ for some $\tau_1\in \Psi(S'',T'')$. Let $\mu'$ be a substitution such that $r_1\mu'\approxOf{\tiny abs} y_1\theta''\tau_1$ and $R_2\subseteq \mathcal{V}(r)$ such that $R_2\cap \mathcal{V}(r_1) = \emptyset$. Using $R_2$ we define a bijective renaming $\nu$ such that for all  $z\in R_2$, $z\nu\not\in \mathcal{V}(r_1\mu')\cup \mathcal{V}(r_1)$. 

We will now consider the term $r\nu\mu'=f(r_1\mu',r_2\nu\mu')$. Note that for all variables $z\in \mathcal{V}(r_1)\cap \mathcal{V}(r_2\nu)$, it must be the case that $z\mu' \genOf{\Abs}  z\mu^*$ where $r_1\mu^* \approxOf{abs} s_1$ and $r_2\mu^* \approxOf{abs} s_2$. Thus, observe that $r_2\nu\mu'\genOf{\Abs}  s_2$.

Now let $\gamma'$ be a substitution such that $\Dom(\gamma')=\mathcal{V}(r_2\nu\mu')$, $r_2\nu
\mu'\gamma'\approxOf{abs} s_2$, and $r_1\mu'\gamma'\approxOf{\Abs} s_1$.  Now consider  $R_2'= \{z\mid z\in \Dom(\gamma')\wedge z\not \in  \mathcal{V}(r_1\mu')\}$ and $\nu' =\{ z\mapsto l\mid z\in R_2'\wedge z\gamma'=l\}$. Note that $r_2\nu\mu'\nu'\preceq_{abs} s_2$ and there exists $t^* \in \Abst_{y_2}(T'',S'')$ such that  $r_2\nu\mu'\nu'\approxOf{\Abs} t^*$ by the definition of the \textit{abstraction set}. For terms in $\Abst_{y_2}(T'',S'')$ we know how to build a $\tau_2\in \Psi(T'',S'')$.

Now let $\mu'_1$ and $\mu'_2$ be substitutions such that $r_1\mu'\approxOf{\Abs} r_1'\mu'_1\mu'_2$ and for all $z\in \Dom(\mu'_2)$ there exists $y\in \Dom(\tau_1)$ such that $z\mu'_2\approxOf{\Abs} y\tau_1$. Notice we can apply the same rewriting to $r_2\nu\mu'\nu'$, that is $r_2'\mu''_1\mu''_2\approxOf{\Abs}r_2\nu\mu'\nu'$. We are free to choose the  $\Dom(\nu')$ such that it does not compose with the range of $\mu'$. Thus for variables $z\in \mathcal{V}(r_1'\mu'_1) \cap \mathcal{V}(r_2'\mu''_1)$ such that $z\in \Dom(\mu_2'')$,  there exists $y\in \Dom(\tau_2)$ such that $z\mu''_2\approxOf{\Abs}y\tau_2$ and $z\mu'_2\approxOf{\Abs}y\tau_1$. We can safely assume that the $\Dom(\tau_2)\cap \mathcal{V}(Ran(\tau_1))= \emptyset$, thus we can choose $\tau\in \Psi(T'',S'')$ such that $\tau=\tau_1\tau_2$ as the required substitution. That is $r\genOf{\Abs}  f(y_1,y_2)\theta''\tau$. \qed

  \end{enumerate}

\end{proof}

Given the technical nature of Theorem~\ref{Thm:Completeness1}, we provide examples in Appendices ~\ref{ex:com1step} and \ref{ex:com2step} illustrating how to produce the generalizations required by the stepcase.

Finally, we remark why we do not merge abstractions in $T$. Note, for a final configuration $\langle \emptyset; S; T;\theta\rangle$, 
there does not exist AUEs $s_1\labTriEq{x} t_1,s_2\labTriEq{y} t_2\in S$, such that  $s_1=s_2$, $t_1=t_2$, and $y\not = x$. This property is not guaranteed for $T$ as we do not have a merge rule for AUEs in $T$. Adding such a rule to $\AU$ may lead to the incompleteness of the procedure:
\begin{example}
Consider the AUE  $g(\varepsilon_f,\varepsilon_f,f(a,h(\varepsilon_f)))\labTriEq{x} g(f(h(\varepsilon_f),a),f(h(\varepsilon_f),$ $a),\varepsilon_f)$, abbreviated as $s\labTriEq{x} t$. Applying $\AU$ to $s\labTriEq{x} t$ results in a set of final configurations containing a configuration where $S= \{\varepsilon_f\labTriEq{y}a,a\labTriEq{z}\varepsilon_f\}$, $T =\{\star\labTriEq{u}h(\varepsilon_f), \star\labTriEq{v}h(\varepsilon_f), h(\varepsilon_f)\labTriEq{w}\star\}$, and $\theta = \{ x\mapsto g(f( u,y),f(v,y),f(z,w))\}$. Furthermore, $\tau=\{u\mapsto h(f(z,y)), w\mapsto h(f(y,z)), v\mapsto h(\varepsilon_f)\}\in \Psi(T,S)$. The generalization induced by this choice of $\tau$ is $ g(f(h(f(z,y)),y),f(h(\varepsilon_f),y),f(z,$ $h(f(y,z))))$. Now consider $T' = \{\star\labTriEq{u}h(\varepsilon_f), h(\varepsilon_f)\labTriEq{w}\star\}$ and $\mu = \{ v\mapsto u\}$. Because $\tau$ maps $u$ and $v$ to different terms there does not exist $\tau'\in \Psi(T',S)$ such that $x\theta\tau\approxOf{\Abs}x\theta\mu\tau'$ as $\mu$ replaces occurrences of $v$ by $u$ in $\theta$ and $\tau'$ must replace occurrences of $\mu$ by the same term. Furthermore, substitution into $y$ and $z$ results in a non-generalization. 
\end{example}

\section{Anti-Unification Type}
\label{sec:AUtype}
This section shows that the complete set of generalizations produced by the $\AU$ algorithm is minimal. We do so by considering a further transformation of the set of final configurations and then show that generalizations constructable from this set of final configurations are incomparable. 
\begin{definition}[Merged configurations]
Let $s$ and $t$ be terms. We refer to  $\AU(\langle \{s\labTriEq{x} t\};\emptyset ; \emptyset; \iota\rangle )$ as \emph{merged} if for all $\langle \emptyset; S_0 ; T_0; \theta_0 \rangle,\langle \emptyset; S_1 ; T_1; \theta_1 \rangle\in \AU(\langle \{s\labTriEq{x} t\};\emptyset ; \emptyset; \iota\rangle )$, $s'\labTriEq{y_1}t'\in S_0$ and $s'\labTriEq{y_2}t'\in S_1$ iff $y_1 = y_2$.
\end{definition}
A merged set of final configurations can be obtained by an appropriate renaming of the store labels 
and applying this renaming to the final substitutions. 
\begin{lemma}
\label{lem:minvars}
    Let $s$ and $t$ be terms and  $\langle \emptyset; S ; T; \theta \rangle\in \AU(\langle \{s\labTriEq{x} t\};\emptyset ; \emptyset; \iota\rangle )$. 
    %
%
    Then for all $s'\labTriEq{y} t'\in S$ and any non-variable term $r$, $x\theta\{y\mapsto r\}\notin \mathcal{G}_{\tiny \Abs}(s,t)$.

\end{lemma}
\vspace{-1em}
\begin{proof}
Given that $s'\labTriEq{y} t'\in S$, we know that  $head(s')\not = head(t')$ and, $head(s')$ and  $head(t')$ are not related absorption symbols. In $x\theta\{y\mapsto r\}$, the non-variable term $r$ replaces $y$ which was a generalization of $s'$ and $t'$, but by this replacement, $head(r)$ will clash with $head(s')$, $head(t')$, or both. Hence, it cannot be a generalization of $s'$ and $t'$, which implies  $x\theta\{y\mapsto r\}\notin \mathcal{G}_{\tiny \Abs}(s,t)$\qed
\end{proof}
\begin{definition}
Let $s$ and $t$ be terms and $\AU(\langle \{s\labTriEq{x} t\};\emptyset ; \emptyset; \iota\rangle)$ merged. We define the set $\mathcal{C}_{\AU}(s,t)$ as  $\mathcal{C}_{\AU}(s,t) = \{ x\theta\tau\mid  \langle \emptyset; S ; T; \theta \rangle\in \AU(\langle \{s\labTriEq{x} t\};\emptyset ; \emptyset; \iota\rangle ) \wedge \tau\in\Psi(T,S)\}.$
\end{definition}
\begin{lemma}
\label{lem:completenessCset}
For any  $s,t$, $\mathcal{C}_{\AU}(s,t)$ is their complete set of $\Abs$-ge\-ne\-ra\-li\-zations.
\end{lemma}
\vspace{-1em}
\begin{proof}
    The lemma follows from completeness of $\AU$ (Theorem~\ref{Thm:Completeness1}). \qed
\end{proof}
\begin{lemma}
\label{lem:incompar}
   For all terms  $s,t$, and $g_0,g_1\in \mathcal{C}_{\AU}(s,t)$, if $g_0\neq g_1$ then neither $g_0 \preceq_{\tiny{\Abs}} g_1$ nor $g_1 \preceq_{\tiny{\Abs}} g_0$ holds. 
\end{lemma}
\vspace{-1em}
\begin{proof}
By Corollary~\ref{cor:rangebnd},  $\mathcal{V}(g_0) \subseteq \mathit{label}(S_0)$ and $\mathcal{V}(g_1)\subseteq \mathit{label}(S_1)$ for some final configurations $\langle \emptyset; S_0 ; T_0; \theta_0 \rangle,\langle \emptyset; S_1 ; T_1; \theta_1 \rangle$ $\in \AU(\langle \{s\labTriEq{x} t\};\emptyset ; \emptyset; \iota\rangle )$. 
By Lemma~\ref{lem:minvars}, w.l.o.g., for $x\in \mathit{label}(S_0)$ we have $g_0\{x\mapsto r\}\not\in \mathcal{G}_{\tiny \Abs}(s,t)$ when $r$ is not a variable. If $r$ is a variable and $r\in \mathit{label}(S_0)\cup \mathit{label}(S_1)$, then $g_0\{x\mapsto r\}\not\in \mathcal{G}_{\tiny \Abs}(s,t)$ because labels in $\mathit{label}(S_0)\cup \mathit{label}(S_1)$ are assigned to unique AUEs (due to merging the results of $\AU$) and thus $x$ and $r$ generalize different terms. Thus,  $r\not \in \mathit{label}(S_0)\cup \mathit{label}(S_1)$ implying neither $g_0 \preceq_{\tiny{\Abs}} g_1$ nor $g_1 \preceq_{\tiny{\Abs}} g_0$ hold.  \qed
\end{proof}
\begin{theorem}
\label{thm:minimality}
   For all terms $s,t$, $\mathcal{C}_{\AU}(s,t)$ is actually $\mcsgabs(s,t)$.
\end{theorem}
\begin{proof}
    Completeness is shown in Lemma~\ref{lem:completenessCset}. Minimality follows from Lemma~\ref{lem:incompar} and Definition~\ref{def:mcsg}.\qed
\end{proof}
\begin{corollary}
    Anti-unification modulo \Abs\ theories is of type infinitary. 
\end{corollary}
\vspace{-1em}
\begin{proof}
By Theorem~\ref{thm:minimality}, the set of $\Abs$-generalizations computed in Example~\ref{ex:procedure2} is an mcsg, which is infinite since \textbf{Configuration 1} produces infinitely many. \qed
\end{proof}
\vspace{-.3em}

Theorem~\ref{thm:minimality} is in contrast to other known infinitary anti-unification problems such as idempotent anti-unification~\cite{DBLP:journals/tocl/CernaK20} where the algorithm produces a finitely representable complete set of generalizations which we can minimize (that is, the set may contain non-minimal generalizations that are cleaned during minimization). In our case, \AU\ directly gives a finitely represented mcsg.
\vspace{-1em}
\section{Computing Linear $\Abs$-Generalizations}
\label{sect:linear}
Linear generalizations do not contain any generalization variable more than once. To consider such a (practically useful) variant, we should drop the two sources of the duplication of generalization variables: the Merge rule and the computation of abstraction substitutions. Instead of the latter, we should replace $x$ with $s$ in the computed generalizations for each $\star \labTriEq{x} s\in T$ and $s \labTriEq{x} \star\in T$, 
because that gives the most specific generalization when our goal is to avoid variable duplication. These lead to the following observations:

\begin{itemize}
    \item generalizations computed in this way form a minimal complete set of linear $\Abs$-generalizations, which is finite (i.e., the linear variant is finitary);
    \item its cardinality bound is $O(2^{n_s})$ where $n_s$ is the number of occurrences of absorption symbols in the input; the exponential bound is caused by branching when we encounter absorption constant/symbol pairs in the AUEs to be generalized;
    \item each linear $\Abs$-generalization is computed in $O(n)$ steps where $n$ is the input size: it is obvious since at each step (rule application) of the algorithm, the number of symbols in the unsolved AUEs strictly decreases. 
\end{itemize}
\begin{example}
    \label{ex:procedure:linear} Recall the AU problem for $g(\varepsilon_f,f(a,h(\varepsilon_f)))$ and $g(f(h(\varepsilon_f),a),\varepsilon_f)$ from Example~\ref{ex:procedure2}. To compute linear $\Abs$-generalizations for these terms by the above-stated modification of \AU, we still get the same four configurations (note that the Merge rule was not applicable), but each of them produces now a single linear $\Abs$-generalization: 
    \begin{itemize}
      \item $g(f (u_1, a)), f (u_2, h(\varepsilon_f)))$ \quad (from configuration 1),
      \item $g(f (u_1, a)), f (a, v_2))$ \quad (from configuration 2), 
      \item $g(f(h(\varepsilon_f),v_1)),f(u_2,h(\varepsilon_f)))$ \quad (from configuration 3), and 
      \item $g(f(h(\varepsilon_f),v_1)),f(a,v_2))$ \quad (from configuration 4).
    \end{itemize}

    These four terms form the minimal complete set of linear $\Abs$-generalizations of $g(\varepsilon_f,f(a,h(\varepsilon_f)))$ and $g(f(h(\varepsilon_f),a),\varepsilon_f)$.

\end{example}

\section{Conclusion}
We introduced a rule-based algorithm that computes generalizations for problems modulo absorption operators and proved that it is sound and complete. Furthermore, the algorithm finitely computes a finite set of final configurations from which we can extract a minimal complete set of generalizations. This set can be infinite for some input, implying that $\Abs$-anti-unification is of type infinitary. We also considered the linear case, discussed the necessary modifications to the algorithm, and showed that the type reduces to finitary.

For future work, we will consider further improvements of the algorithmic techniques for generating minimal complete sets of generalizations and how to combine our algorithm with algorithms for computing generalizations over other equational theories, similar to the analysis performed in~\cite{DBLP:journals/amai/AlpuenteEMS22}. One considered improvement to the algorithm would be a grammatical representation of the computed mcsg. Additionally, we plan to consider how such equational theories can be used in practice as part of methods for software analysis.

\vspace{-2mm}

\subsubsection*{Acknowledgments.} 
This work was supported by the Czech Science Foundation Grant No. 22-06414L; the Austrian Science Fund (FWF) project P 35530; Cost Action CA20111 EuroProofNet; the Brazilian agency CNPq, Grant Universal 409003/21-2, and RG 313290/21-0; and the Brazilian Federal District Research Foundation  FAPDF, Grant DE 00193-00001175/2021-11.  The Brazilian Higher Education Council (CAPES) supported the Brazilian-Austrian cooperation through the program PrInt.
\bibliographystyle{plain}
\bibliography{ref}
\appendix

\section{Stepcase, Case 1 Example}
\label{ex:com1step}
 The following examples walk through the steps presented in the proof of Theorem~\ref{Thm:Completeness1}, \textbf{Stepcase}, case 1.
\begin{example}
\label{ex:compl1}
     Let us consider the AUE $t_1 \labTriEq{x} t_2$ presented in Example~\ref{ex:procedure1} where
    $t_1=g(\varepsilon_f,f(a,h(\varepsilon_f))),$ $ t_2=g(f(h(\varepsilon_f),a),\varepsilon_f),$
and the final configuration 
\begin{align*}
    \langle \emptyset; \{\varepsilon_f\labTriEq{v_1} a, a\labTriEq{u_2} \varepsilon_f\};\{\star\labTriEq{u_1} h(\varepsilon_f),h(\varepsilon_f)\labTriEq{v_2}\! \star\}; \{x\mapsto g(f(u_1,v_1),f(u_2,v_2))\}\rangle
\end{align*}
    derived from this AUE. Also, the term $r$ is a generalization of the above AUE:
    $$r=g(\overbrace{z}^{r_1},\overbrace{f(y,h(z)))}^{r_2}=g(\overbrace{w_1}^{r_1'},\overbrace{f(w_2,w_3)}^{r_2'})\overbrace{\{w_1\mapsto z, w_2 \mapsto y, w_3\mapsto h(z)\}}^{\mu} $$
    Notice that $z\in \mathcal{V}(r_1)\cap\mathcal{V}(r_2)$ and there is a position $1\in pos(t_1)\cap pos(t_2)$ such that $z$ generalizes $t_1\vert_1=\varepsilon_f$ and $t_2\vert_1=f(h(\varepsilon_f),a)$. 
    
    Consider the final configuration $\langle \emptyset;\{\varepsilon_f\labTriEq{y'} a\};\{\star\labTriEq{z'}h(\varepsilon_f)\};\{x'\mapsto f(y',z')\} \rangle $ $\in \AU(\langle \{ \varepsilon_f\labTriEq{x'}f(h(\varepsilon_f),a)\};$ $\emptyset;\emptyset; \theta_z \rangle)$. A possible choice for $\tau^*$ is $\tau^*=\{y'\mapsto h(\varepsilon_f)\}$, and $x'\{x'\mapsto f(y',z')\}\tau^*=f(h(\varepsilon_f),z')$. 
And we have:
\[\{x\mapsto g(f(u_1,v_1),f(u_2,v_2))\} =\{x\mapsto g(y_1,y_2)\} \overbrace{\{y_1\mapsto f(u_1,v_1),y_2\mapsto f(u_2,v_2)\}}^{\theta''},\]
and, in particular,  $g(w_1,f(w_2,w_3))\genOf{\Abs}  g(y_1,y_2)\theta''$.  Notice that the following statements hold: 

    $$r_1'\mu \genOf{\Abs}  y_1\theta''\overbrace{\{u_1\mapsto  h(\varepsilon_f)\}}^{\tau_1}=f(h(\varepsilon_f),v_1)$$
   $$r_2'\mu \genOf{\Abs}  y_2\theta''\overbrace{\{v_2\mapsto h(f(h(\varepsilon_f),v_1))\}}^{\tau_2}=f(u_2,h(f(h(\varepsilon_f),v_1)))$$
    where  $\gamma_1 = \{z\mapsto f(h(\varepsilon_f),v_1)\}$ and  $\gamma_2 = \{y\mapsto u_2, z\mapsto f(h(\varepsilon_f),v_1)\}$ realize these comparisons, i.e. $r_1'\mu \gamma_1\approx_{\tiny\Abs} y_1\theta''\tau_1$ and $r_2'\mu \gamma_2\approx_{\tiny\Abs} y_2\theta''\tau_2$. Hence, we care only about $(\mu\gamma_1) \vert_{\mathcal{V}(r'_1)}$ and  $(\mu\gamma_2)\vert_{\mathcal{V}(r'_1)}$, which can be rewritten as follows 
    
    \scalebox{.88}{\begin{minipage}{\textwidth}$$\mu\gamma_1\vert_{\mathcal{V}(r'_1)} = (\overbrace{\{ w_1\mapsto f(q_1,v_1)\}}^{\mu_1^1} \overbrace{\{q_1\mapsto h(\varepsilon_f)\}}^{\mu_1^2})\mid_{\mathcal{V}(r'_1)},$$ $$
    \mu\gamma_2\mid_{\mathcal{V}(r'_2)} = (\overbrace{\{ w_2\mapsto u_2, w_3\mapsto q_2\}}^{\mu_2^1}\overbrace{\{q_2\mapsto h(f(h(\varepsilon_f),v_1))\}}^{\mu_2^2})\mid_{\mathcal{V}(r'_2)}.$$
   \end{minipage}}
   \vspace{1em}
   
   \noindent We define $\tau$ such that $u_1\tau=q_1\mu_1^2$ and $v_2\tau=q_2\mu_2^2$, that is
   $$\tau = \{ u_1\mapsto h(\varepsilon_f), v_2\mapsto h(f(h(\varepsilon_f),v_1))\}.$$
    Thus, the result is $g(w_1,f(w_2,w_3)\mu\genOf{\Abs} g(y_1,y_2)\theta''\tau$, that is 
    $$g(z,f(y,h(z)))\genOf{\Abs} g(f(h(\varepsilon_f),v_1),f(u_2,h(f(h(\varepsilon_f),v_1)))).$$

\end{example}

\section{Stepcase, Case 2 Example}
 The following examples walk through the steps presented in the proof of Theorem~\ref{Thm:Completeness1}, \textbf{Stepcase}, case 2.
\label{ex:com2step}
 \begin{example}  Let us consider the AUE 

$$\varepsilon_f\labTriEq{x}f(f(h(b,a),a),f(b,h(a,b))))$$ 
and the final configuration below derived from this AUE. 
$$\langle\emptyset; \{\varepsilon_f\labTriEq{u_2}a\};\{\star\labTriEq{y_2}f(b,h(a,b)), \star\labTriEq{u_1}h(b,a)\};$$ $$\{x\mapsto f(f(u_1,u_2), y_2), y_1\mapsto f(u_1,u_2)\}\rangle$$
The next term is a generalization of the initial AUE. $$r=f(\overbrace{f(w,u)}^{r_1},\overbrace{f(y,h(u,z))}^{r_2})$$
Using the abstractions sets:
$$\Abst_{y_2}(f(b,h(a,b)),\{u_2\mapsto a\})=\{f(b,h(a,b)),f(b,h(u_2,b))\}$$
$$\Abst_{u_1}(h(b,a),\{u_2\mapsto a\})=\{h(b,a),h(b,u_2)\}$$
we get
$$r_1\genOf{\Abs} y_1\overbrace{\{y_1\mapsto f(u_1,u_2)\}}^{\theta''}\overbrace{\{ u_1\mapsto h(b,u_2)\}}^{\tau_1}=f(h(b,u_2),u_2)$$ 
Thus, $\mu'=\{w\mapsto h(b,u_2), u\mapsto u_2\}$. Notice that $u\in \mathcal{V}(r_2)$, hence $\nu=\{y\mapsto y', z\mapsto z'\}$. Applying both substitutions to $r$:
$$r\nu\mu'=f(\overbrace{f(h(b,u_2),u_2)}^{r_1\nu\mu'},\overbrace{f(y',h(u_2,z'))}^{r_2\nu\mu'}$$
Notice that $r_2\nu\mu'\genOf{\Abs} f(b,h(a,b))$. One possible choice for $\gamma'$ is $\{y'\mapsto b, u_2\mapsto a, z'\mapsto b\}$. Notice that $y',z'\in R'_2$ and thus $\nu'= \{y'\mapsto b, z'\mapsto b \}$. This leaves us with $r_2\nu\mu'\nu'\approxOf{\Abs} f(b,h(a,b))$ and $\tau_2 =\{y_2\mapsto h(b,h(u_2,b))\}$. Composing the two substitutions results in 
$$ \tau =\tau_1\tau_2 =\{u_1\mapsto h(b,u_2), y_2\mapsto f(b,h(u_2,b))\}$$
and 
$$f(f(w,u),f(y,h(u,z)))\genOf{\Abs} f(f(h(b,u_2),u_2),f(b,h(u_2,b))).$$
\end{example}

\begin{example} Consider the problem below. 

$$g(\varepsilon_f,a)\labTriEq{x_s}g(f(f(a,h(\varepsilon_f)),h(\varepsilon_f)),\varepsilon_f)$$ 
We can follow the next branch, which leads us to the problem\\ $(\varepsilon_f\labTriEq{x}f(f(a,h(\varepsilon_f)),h(\varepsilon_f)))$:
$$\langle \{\varepsilon_f\labTriEq{x}f(f(a,h(\varepsilon_f)),h(\varepsilon_f)), a \labTriEq{y}{\triangleq}\varepsilon_f\};\emptyset;\emptyset;\{x_s\mapsto g(x,y)\} \rangle$$
$$\langle \{\varepsilon_f\labTriEq{y_1}{\triangleq}f(a,h(\varepsilon_f)), a \labTriEq{y}\varepsilon_f\};\emptyset;\{\star \labTriEq{y_2}h(\varepsilon_f)\};$$ $$\{x_s\mapsto g(f(y_1,y_2),y),  x\mapsto f(y_1,y_2)\} \rangle$$

and the final configuration 
\begin{align*}
        \langle \emptyset;\{\varepsilon_f\labTriEq{u_1}a, a \labTriEq{y}\varepsilon_f\};\{\star \labTriEq{y_2}h(\varepsilon_f), \star\labTriEq{u_2}{\triangleq}h(\varepsilon_f)\}; & \\
         \{x_s\mapsto g(f(f(u_1,u_2),y_2),y), x\mapsto f(f(u_1,u_2),y_2), y_1\mapsto f(u_1,u_2)\} \rangle &
    \end{align*}

derived from this AUE. And the next is a generalization of the initial AUE 
$$r=f(\overbrace{f(w_1,h(f(w_3,z)))}^{r_1},\overbrace{h(f(w_1,f(w_2,z)))}^{r_2})$$
Using the abstractions sets:
$$\Abst_{y_2}(h(\varepsilon_f),\{u_1\mapsto a, y\mapsto \varepsilon_f\})=\{h(\varepsilon_f),h(f(y,u_1)),h(f(a,f(u_1,y))), \ldots\}$$
$$\Abst_{u_2}(h(\varepsilon_f),\{u_1\mapsto a, y\mapsto \varepsilon_f\})=\Abst_{y_2}$$
we get
$$r_1\genOf{\Abs} y_1\overbrace{\{y_1\mapsto f(u_1,u_2)\}}^{\theta''}\overbrace{\{ y_2\mapsto h(f(a,f(u_1,y))\}}^{\tau_1}=f(u_1,h(f(a,f(u_1,y))))$$ 
Thus, $\mu'=\{w_1\mapsto u_1, w_3\mapsto a, z\mapsto f(u_1,y)\}$. Notice that $z\in \mathcal{V}(r_2)$, thus $\nu=\{w_2\mapsto w'\}$. Applying both substitutions to $r$:
$$r\nu\mu'=f(\overbrace{f(u_1,h(f(a,y)))}^{r_1\nu\mu'},\overbrace{h(f(u_1,f(w',f(u_1,y))))}^{r_2\nu\mu'}$$
Notice that $r_2\nu\mu'\genOf{\Abs} h(\varepsilon_f)$. One possible choice for $\gamma'$ is $\{y\mapsto \varepsilon_f , u_1\mapsto \varepsilon_f, w'\mapsto \varepsilon_f\}$. Notice that $w'\in R'_2$ and thus $\nu'= \{w'\mapsto \varepsilon_f\}$. This leaves us with $r_2\nu\mu'\nu'\approxOf{\Abs} f(b,h())$ and $\tau_2 =\{y_2\mapsto h(b,h(u_2,b))\}$. Composing the two substitutions results in 
$$ \tau =\tau_1\tau_2 =\{u_1\mapsto h(u_2,b), y_2\mapsto f(b,h(u_2,b))\}$$
and 
$$f(f(w,u),f(y,h(u,z)))\genOf{\Abs} f(f(h(u_2,b),u),f(b,h(u_2,b))).$$
\end{example}

\end{document}